\documentclass[sigconf]{aamas} 
\usepackage{balance} % for balancing columns on the final page

\renewcommand\footnotetextcopyrightpermission[1]{} % removes footnote with conference information in first column
\usepackage{times}
\usepackage{soul}
\usepackage{url}
\usepackage[utf8]{inputenc}
\usepackage{caption}
\usepackage{graphicx}
\usepackage{amsmath}
\usepackage{amsthm}
\usepackage{booktabs}
\usepackage[normalem]{ulem}
\usepackage{algorithm}
\usepackage{algorithmic}

\usepackage{listings}
\DeclareCaptionStyle{ruled}{labelfont=normalfont,labelsep=colon,strut=off} % DO NOT CHANGE THIS
\floatstyle{ruled}
\newfloat{listing}{tb}{lst}{}
\floatname{listing}{Listing}

 \usepackage{amsthm}
\usepackage{thmtools,thm-restate}
\usepackage{float}
\usepackage{graphicx}
\usepackage{soul}
\usepackage{tikz-cd}
\usepackage{tikz}
\usetikzlibrary{positioning,chains,fit,shapes,calc,snakes}
\usepackage{enumerate}

\usepackage{tabularx}
\usepackage{caption}
\usepackage{multirow}
\usepackage{todonotes}
\usepackage{colortbl}
\usepackage{tabulary}
\usepackage{etoolbox}
\usepackage{array,booktabs}
\usepackage{diagbox}
\usepackage{subcaption}
\usepackage{xcolor}
\usepackage{thm-restate}
\usepackage{hyperref}
\hypersetup{
     colorlinks   = true,
     linkcolor    = red, % color of internal links
     urlcolor     = teal, % color of external links
	 citecolor    = blue % color of links to bibliography
}
\usepackage{mathtools}
\usepackage{xspace}
\usepackage[capitalise,noabbrev]{cleveref}
\Crefname{property}{Property}{Properties}
\Crefname{example}{Example}{Examples}
\Crefname{table}{Table}{Tables}
\Crefname{figure}{Figure}{Figures}

\newtheorem{proposition}{Proposition}

\newtheorem{example}{Example}

\newtheorem{lemma}{Lemma}
\newtheorem{definition}{Definition}

\newtheorem{remark}{Remark}

\newcommand{\enn}{\hat{n}}
\newcommand{\emm}{\hat{m}}
\DeclareMathOperator*{\argmax}{argmax}

\newcommand{\SR}[1]{}

 \newcommand{\memw}{min \#envy max \sw{}\xspace}
 \newcommand{\mtemw}{min total envy max \sw{}\xspace}
 \newcommand{\mtec}{min total envy complete\xspace}
 \newcommand{\mec}{min \#envy complete\xspace}
 \newcommand{\mcmw}{{minimum cost perfect matching}\xspace}
\newcommand{\msEF}{maximum size EF\xspace}

\newcommand{\egalk}{max \egwel{}\xspace}
\newcommand{\egaln}{max \egwel{}\xspace}

\newcommand{\minEnvy}{min \#envy\xspace}
\newcommand{\minTEnvy}{min total envy\xspace}

\newcommand{\minmaxTEnvy}{minimax total envy\xspace}
 
 \newcommand{\NPc}{\textrm{\textup{NP-c}}}
\newcommand{\NPC}{\textrm{\textup{NP-complete}}}

\newcommand{\NPH}{\textrm{\textup{NP-hard}}}
\newcommand{\WH}{\textrm{\textup{W[1]-hard}}}
 \newcommand{\Poly}{\textup{P}}

\colorlet{mycolor}{gray!25}

\newcommand{\EF}[1]{\ifstrempty{#1}{\textrm{\textup{EF}}}{\textrm{\textup{EF{$#1$}}}}}

\newcommand{\envy}{\ensuremath{\mathsf{envy}}}
\newcommand{\tenvy}{\ensuremath{\mathsf{total\textsc{ }envy}}}

\newcommand{\topp}{\ensuremath{\max}}
\newcommand{\sw}{\ensuremath{\mathsf{USW}}}
\newcommand{\egwel}{\ensuremath{\mathsf{ESW}}}

\acmSubmissionID{008}

\title{The Degree of Fairness in Efficient House Allocation}
\author{Hadi Hosseini}
\affiliation{
  \institution{Pennsylvania State University}
  \country{University Park, USA}
  }
\email{hadi@psu.edu}

\author{Medha Kumar}
\affiliation{
   \institution{Pennsylvania State University}
  \country{University Park, USA}
}
\email{mkumar@psu.edu}

\author{Sanjukta Roy}
\affiliation{
  \institution{University of Leeds}
  \country{Leeds, UK}
  }
\email{s.roy@leeds.ac.uk}

\begin{abstract}
    
      The classic house allocation problem is primarily concerned with finding a matching between a set of agents and a set of houses that guarantees some notion of economic efficiency (e.g. utilitarian welfare). While recent works have shifted focus on achieving fairness (e.g. minimizing the number of envious agents), they often come with notable costs on efficiency notions such as utilitarian or egalitarian welfare.
      We investigate the trade-offs between these welfare measures and several natural fairness measures that rely on the number of envious agents, the total (aggregate) envy of all agents, and maximum total envy of an agent.
     In particular, by focusing on envy-free allocations, we first show that, should one exist, finding an envy-free allocation with maximum utilitarian or egalitarian welfare is computationally tractable.
      We highlight a rather stark contrast between utilitarian and egalitarian welfare by showing that finding utilitarian welfare maximizing allocations that minimize the aforementioned fairness measures can be done in polynomial time while their egalitarian counterparts remain intractable (for the most part) even under binary valuations.
      We complement our theoretical findings by giving insights into the relationship between the different fairness measures and conducting empirical analysis.
\end{abstract}

\keywords{Fairness, Welfare, Utilitarian, Egalitarian, House Allocation}

\newcommand{\BibTeX}{\rm B\kern-.05em{\sc i\kern-.025em b}\kern-.08em\TeX}

\begin{document}
%%% The following commands remove the headers in your paper. For final 
%%% papers, these will be inserted during the pagination process.

\pagestyle{fancy}
\fancyhead{}

%%% The next command prints the information defined in the preamble.

\maketitle

\section{Introduction}
 
The classic house allocation problem is primarily concerned with assigning a set of houses (or resources) to a set of agents based on their preferences over houses such that each agent receives at most one house. 
It was motivated by a variety of applications such as kidney exchange \cite{roth2004kidney} or labour market \cite{hylland1979efficient} 
where agents are initially endowed with houses or houses have to be distributed afresh among agents.\footnote{This problem is commonly known as ``Shapley-Scarf Housing Markets'' when agents are initially endowed with houses. The goal is often finding mutually-beneficial exchanges that lead to efficient stable allocations (see \cite{shapley1974cores,roth1982incentive}).} 
%
% This model was motivated by a variety of application in dormitory assignment, medical resource allocation, and kidney exchange.
While this model was primarily studied for designing incentive compatible mechanisms \cite{svensson1999strategy,abdulkadirouglu2003school} along with some notion of economic efficiency, recent works have shifted focus to the issues of fairness such as envy-freeness (EF), which requires that every agent weakly prefers its allocated house to that of every other agent. An envy-free allocation may not always exist: for example, consider two agents who both like the same house. The fair division literature contains a variety of \textit{approximate} envy measures (e.g. envy-free up to one \cite{lipton2004approximately,caragiannis2019unreasonable}) that cannot appropriately be utilized here due to the `one house per agent' constraint.

An orthogonal, but more suitable, approach is measuring the `degree of envy' \cite{chevaleyre2006issues} among the agents by counting the number of envious agents or the total aggregate envy experienced. In this vein, recent works have investigated the existence of envy-free house allocations under ordinal preferences \cite{GSV19envy}, maximum-size envy-free allocation \cite{aigner2022envy}, and complete allocations that minimizes the number of envious agents \cite{KMS21complexity,MadathilMS23} or those that minimize the total aggregate envy among agents \cite{hosseini2023graphical}.
Yet, these approaches often take a toll on efficiency notions such as size (the number of allocated agents), utilitarian welfare (the sum of agents' values), or egalitarian welfare (the value of the worst off agent).  The following example illustrates some of these nuances.

\begin{example}[\textbf{Fairness in House Allocation}] 
\label{ex:tradeoff}
\begin{figure}[t] 
    \centering
    \begin{minipage}{0.45\linewidth}
        \centering
         \resizebox{!}{0.9\linewidth}{
        {
        \begin{tikzpicture}[every node/.style={draw,circle}, fsnode/.style={fill=black}, ssnode/.style={fill=black}]
        
        % the agents
        \begin{scope}[start chain=going below,node distance=5mm]
        \foreach \i in {$a_1$, $a_2$, $a_3$, $a_4$}
          \node[fsnode,on chain] (f\i) [label=left: \i] {};
        \end{scope}
        
        % the houses
        \begin{scope}[xshift=3cm,yshift=0.35cm,start chain=going below,node distance=5mm]
        \foreach \i in {$h_1$, $h_2$, $h_3$, $h_4$, $h_5$}
          \node[ssnode,on chain] (s\i) [label=right: \i] {};
        \end{scope}
        
        % the edges
        \tikzset{decoration={snake,amplitude=.7mm,segment length=4mm,
                       post length=0mm,pre length=0mm}}
        \draw[ thick, draw=black] (f$a_1$) -- (s$h_1$);
        \draw[ thick, draw=black] (f$a_1$) -- (s$h_2$);
        % \draw (f$a_2$) -- (s$h_1$);
        \draw[ thick, draw=black] (f$a_2$) -- (s$h_2$);
        \draw[ thick, draw=black] (f$a_3$) -- (s$h_1$);
        \draw[ thick, draw=black] (f$a_4$) -- (s$h_2$);
        % \draw (f$a_5$) -- (s$h_3$);
        \draw[ blue!70, line width=4, opacity=0.4] (f$a_1$) -- (s$h_1$);
        \draw[ blue!70, line width=4, opacity=0.4] (f$a_2$) -- (s$h_2$);
        \draw[decorate, orange, line width=0.75] (f$a_1$) -- (s$h_3$);
        \draw[decorate, orange, line width=0.75] (f$a_2$) -- (s$h_4$);
        \draw[decorate, orange, line width=0.75] (f$a_3$) -- (s$h_5$);
        
        \tikzset{decoration={snake,amplitude=.7mm,segment length=6mm,
                      post length=0mm,pre length=0mm}}
        \draw[dashed, green, line width=1] (f$a_1$) to[out=30,in=160] (s$h_1$);
        \draw[dashed, green, line width=1] (f$a_2$) -- (s$h_3$);
        \draw[dashed, green, line width=1] (f$a_3$) -- (s$h_4$);
        \draw[dashed, green, line width=1] (f$a_4$) -- (s$h_5$);
       \end{tikzpicture}}
           }
        \end{minipage}\hfill
        \begin{minipage}{0.5\linewidth}
        \caption{\small{An envy-free allocation of maximum size is shown in wavy orange; a minimum \#envy complete allocation is shown in dashed green; no envy-free allocation has the maximum welfare $2$; and the allocation in blue achieves  minimum \#envy among the matchings that has the maximum welfare. %a minimum envy maximum welfare allocation is shown in blue.
        }}
     \label{fig:mainexample}
    \end{minipage}

  % \vspace{-15pt}
%
\end{figure}
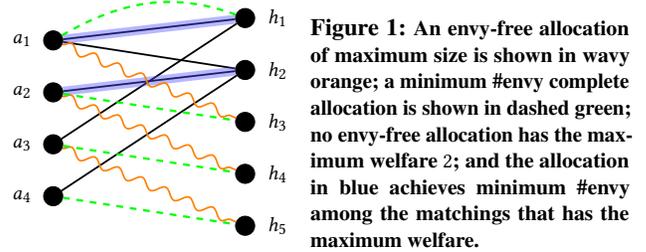
%\vspace{-15pt}
%
% \end{example}
Consider an instance with four agents $\{a_1,a_2,a_3,a_4\}$ and five houses $\{h_1, h_2, \dots, h_5\}$ and binary valuations. For the ease of exposition, we use a graphical representation of the problem as shown in \cref{fig:mainexample}, where solid lines indicate an agent has a positive valuation for the house.

A maximum-size envy-free allocation (shown in orange) assigns houses $h_3$, $h_4$, and $h_5$ to agents $a_1$, $a_2$, and $a_3$ respectively.
 A maximum utilitarian welfare allocation that minimizes the number of envious agents (shown in blue) has two envious agents $a_3$ and $a_4$, and a utilitarian welfare of two. There is no envy-free allocation with maximum utilitarian welfare.
A complete allocation (agent saturating) that minimizes the number of envious agents is shown in green. 
While the size of the maximum size envy-free allocation is three, both the maximum size envy-free allocation and the envy-free allocation with maximum utilitarian welfare, have zero utilitarian welfare.
\end{example}

Finding a complete allocation that minimizes the number of envious agents (the green allocation in the above example) is shown to be \NPC{} \cite{KMS21complexity} even for binary valuations. Moreover, a maximum-size envy-free allocation proposed by \cite{aigner2022envy} only considers allocations to positive valued houses and returns `empty' otherwise, and the envy-free algorithm proposed by \cite{GSV19envy} returns `none' because there is no complete envy-free allocation that assigns \textit{every} agent to a house it likes.
While these observations illustrate the intricate trade-offs between each approach, the relation between these notions and their computational aspects remain unclear.
\begin{table}[t] 
\centering 
\resizebox{\linewidth}{!}{%
\begin{tabular}{@{}llll@{}}
               \cmidrule(l){2-4} 
               & max \sw{} & max \egwel{}  & allocation size $\geq k$ \\ \midrule
envy-free (\EF{})     & \Poly{} (Prop.~\ref{thm:WeightedEFmaxWelfare})$^\dagger$       &   \Poly{} (Thm.~\ref{thm:EFegalwel})       &  P (Prop.~\ref{thm:WeightedEFmaxsize})$^\star$        %\HH{shouldn't this be Prop.~\ref{thm:WeightedEFmaxWelfare}}
\\
% \rowcolor[HTML]{CBCEFB} 
\midrule
min \#envy       &      \Poly{}  (Thm.~\ref{thm:WeightedminEmaxWelfare})  &  \NPc{} (Thm.~\ref{thm:minenvyk-egal})   & \begin{tabular}[c]{@{}l@{}}$m > n$: \NPc{}$^\mathsection$\\ $m \leq n$: \Poly{} (Thm.~\ref{thm:WeightedminEnvycomplete})\end{tabular} \\
\midrule
min total envy &     \Poly{} (Thm.~\ref{thm:minTotalEnvyMaxWelfare})     &\begin{tabular}[c]{@{}l@{}}  as hard as (\$)  \\ %\HH{or equivalence?} 
(Thm.~\ref{thm:redfromMTEC}) \end{tabular}  &  \begin{tabular}[c]{@{}l@{}} $m > n$: open (\$)  \\ $m \leq n$: \Poly{} (Cor.~\ref{cor:minTotalEnvyComplete-mleqn} )$^\ddagger$\end{tabular}                 \\ 
%\midrule
%minimax envy \HH{affected by Alg 1.}      &      \Poly{} (Thm.~\ref{thm:min-maxenvySW})  &   \Poly{} (Thm.~\ref{thm:minmaxenvy-egal-wel})  & \Poly{} (Thm.~\ref{thm:minimaxEmaxsize}) \\
\midrule
minimax total envy &    open   &  \NPc{}  (Thm.~\ref{thm:mmtenvy-egal-welfare}) & \NPc{}$^\ddagger$            \\ 
\bottomrule
\end{tabular}
}
\caption{\small{
The summary of results for weighted instances with $n$ agents, $m$ houses, and $0\leq k \leq \min\{m,n\}$. \Poly{} and \NPc{} refer to polynomial time and \NPC{}, respectively. (\$) refers to the min total envy complete when $m>n$.
Result(s) marked by $^\mathsection$ is due to \protect\cite{KMS21complexity}, those marked by $\ddagger$ and $^\dagger$ are shown by \protect\cite{MadathilMS23} and \protect\cite{aigner2022envy}, respectively for binary valuations,
and $^\star$ is the result by \protect\cite{GSV19envy} for $k=n$.}
% % $^\mathsection$ \protect\cite{KMS21complexity},
% $^*$ proven for $k=n$ by  \protect\cite{GSV19envy}, and
% $\ddagger$ \protect\cite{MadathilMS23}
}
\label{tab:summary}%\vspace{-20pt}
\end{table}

\subsection{Our Contributions}
We consider three well-studied notions of economic efficiency - the size (i..e the number of allocated agents), the utilitarian welfare, and the egalitarian welfare of the allocation. We study these notions under both binary and arbitrary positive valuations and investigate their interaction with various envy fairness measures. These include envy-freeness (EF), minimizing the number of envious agents (min \#envy), and minimizing the total envy of all agents (min total envy), and minimizing the total envy of the most envious agent (minimax total envy).
% Additionally, we also consider fairness measures for individual agents, i.e., minimizing the total envy of the most envious agent (minimax total envy).
% Two of them ensure global fairness, namely, minimum number of envious agents (min envy) and minimum total (aggregate) envy of all agents (min total envy). The other two accomplish fairness for individual agents, namely, minimizing the envy of the most envious agents (minimax envy) and minimizing total envy of the most envious agent (minimax total envy). 
Table~\ref{tab:summary} summarizes our results.

\paragraph{Envy-free allocations.}
We show that an envy-free allocation of maximum size can be computed efficiently under binary (\cref{prop:EFmaxSize}) or arbitrary non-negative valuations (\cref{thm:WeightedEFmaxsize}).
By focusing on welfare-maximizing allocations, we show that, should one exist, finding an envy-free allocation that maximizes utilitarian welfare (\sw{}) or egalitarian welfare (\egwel{}) is computationally tractable (\cref{thm:WeightedEFmaxWelfare} and \cref{thm:EFegalwel}, resp.).

\paragraph{Utilitarian welfare.}
We show that finding an allocation with min \#envy %that minimizes the number of envious agents
(\cref{thm:WeightedminEmaxWelfare}), or min total envy %minimizes the total envy of all agents
(\cref{thm:minTotalEnvyMaxWelfare}), is tractable under the added constraint of maximizing \sw{}. Without the welfare constraint, the former problem has been proven to be \NPH{} when we aim to find a complete allocation \cite{KMS21complexity}.
%
%Furthermore, we show that finding a minimax envy allocation, while maximizing \sw{}, can be done in polynomial time (\cref{thm:min-maxenvySW}).
%
% Contrary to these results, we show by restricting the search space to allocations that maximize the utilitarian welfare, minimizing the number of envious agents (\cref{thm:WeightedminEmaxWelfare}) or minimizing the total envy among all agents 
% \cref{thm:minTotalEnvyMaxWelfare} become tractable for any arbitrary number of agents and houses.
%
% hat a minimum envy allocation that maximizes the utilitarian welfare (\cref{thm:WeightedminEmaxWelfare}) and  can be computed in polynomial time  for any arbitrary number of agents and houses. 
% Similarly, 
%
Additionally, we analyze the relationship between the size of an allocation and its \sw{}. We obtain polynomial time algorithms for finding min \#envy (\cref{thm:WeightedminEnvycomplete}), min total envy (\cref{cor:minTotalEnvyComplete-mleqn}) complete allocations when $m \leq n$. %, and minimax envy complete allocation (\cref{thm:minimaxEmaxsize} and \cref{cor:mmEcomplete}).

\paragraph{Egalitarian welfare.}
To complement our study, we consider the well-established Rawlsian notion of egalitarian welfare. %, which selects an allocation that maximizes the minimum value of the agents.
We first show that finding an envy-free allocation, when one exists, with maximum \egwel{} can be done in polynomial time (\cref{thm:EFegalwel}).
We highlight a contrast between egalitarian and utilitarian welfare by showing that when considering allocations that maximize the egalitarian welfare (as opposed to utilitarian ones), finding a min \#envy allocation is \NPC{} (\cref{thm:minenvyk-egal}).
%Nonetheless, we show that minimax envy allocations  within all egalitarian welfare maximizing allocations can be computed in polynomial time (\cref{thm:minmaxenvy-egal-wel}). However,
Moreover, it is \NPC{} to find minimax total envy max \egwel{} (\cref{thm:mmtenvy-egal-welfare}). 
%
% It is easy to check that the decision version of min \#envy (or minimax total envy) max \egwel{} where we ask if there exists a max \egwel{} allocation with \#envy (resp. max total envy) at most $t$ is in NP, thus we only show the proof for \NPH ness.\HH{this last part seems too much; move to that section instead of here.}
%
Finally, while the computational complexity of finding a complete allocation that minimizes the total envy remains open, we show an intriguing relation to its egalitarian welfare counterpart (\cref{thm:redfromMTEC}) and conclude with complementary experimental observations under randomly generated valuations.

% In summary, we show that utilitarian welfare maximization prevents agents from getting undesirable houses when better houses are available at the cost of creating envy, resulting in tractable algorithms. A house 
% $h$ is undesirable for agent $i$ if there exists another house $h'$ such that 1) $v_i(h') > v_i(h)$ and 2) $h'$ is not assigned to any other agent.
% Intuitively, if an agent likes a house that is available (not assigned to anyone), then the agent does not care about other houses that are less valued.
% The hardness of the problems primarily comes from deciding how to allocate the undesirable houses to achieve fairness (i.e., by creating ``less'' envy). Thus, the fair house allocation problems remain intractable when undesirable houses can be allocated, e.g., when finding a fair complete allocation and $m>n$, or when finding a fair  egalitarian welfare maximizing allocation. 
On the technical level, the computational hardness often comes from deciding how to utilize the allocation of ``undesirable'' houses to improve fairness (i.e., by reducing envy). A house is undesirable for an agent if there exists another house with higher value which is not assigned to any other agent. Intuitively, if an agent likes a house that is available (not assigned to anyone), then the agent does not care about other houses that are less valued.
Thus, different variants of the fair house allocation problem remain intractable when undesirable houses can be allocated, e.g., when finding a fair complete allocation and $m>n$, or when finding a fair  egalitarian welfare maximizing allocation.

\subsection{Related Work}

Fairness in house allocations is a well studied problem.  House allocations were first studied in social choice where the focus was on finding efficient and strategyproof allocations~\cite{roth1982incentive}. As the focus shifted towards fair allocations, increasingly algorithmic approaches were used. 
Achieving fairness by minimizing the number of envious agents or total envy has been studied previously by~\cite{GSV19envy,aigner2022envy,KMS21complexity,MadathilMS23}. 
\citet{kamiyama2021envy} showed hardness of finding EF solutions for pairwise preferences.
\citet{belahcene2021combining} studied a relaxed notion called ranked envy-freeness. The hardness and approximability of minimizing total envy was studied for housing allocation problem where agents are located on a graph \cite{hosseini2023tight,hosseini2023graphical}.
% \HH{this paragraph may be moved to Related Work.}
The egalitarian allocations have been studied under different names such as the classic makespan minimization problem in job scheduling \cite{lenstra1990approximation}, the Santa Claus problem \cite{bansal2006santa}, and in fair allocation of resources \cite{lipton2004approximately,bouveret2016characterizing}. In these settings the problem of maximizing the egalitarian welfare (worst-off agent) is shown to be \NPH{}~\cite{bouveret2016characterizing}, giving rise to several approximation algorithms \cite{shmoys1993approximation}. See \cref{app:relatedwork} for an extended discussion on related work.
Our setting is crucially different from these works in two  ways: in the house allocation problem each agent receives at most one house, and not all houses need to be assigned. 

\section{Model} \label{sec:model}

An instance of the \emph{house allocation problem} is represented by a tuple $\langle N, H, V \rangle$, where $N \coloneqq \{1, 2, \ldots, n\}$ is a set of $n$ agents, $H \coloneqq \{h_1, h_2, \ldots, h_m\}$ is a set of $m$ houses, and $V \coloneqq (v_{1}, v_{2}, \ldots, v_{n})$ is a \emph{valuation profile}. 
Each $v_i(h)$ indicates agent $i$'s non-negative value for house $h \in H$. Thus, for $i\in N$ the value of a house $h\in H$ is $v_i (h) \geq 0$, and $v_i (\emptyset) = 0$.
An instance is \textit{binary} if for every $i\in N$ and every $h\in H$, $v_{i} (h) \in \{0,1\}$; otherwise it is a \textit{weighted} instance.

An allocation $A$ is an injective mapping from agents in $N$ to houses in $H$. For each agent $i\in N$, $A(i)$ is the house allocated to agent $i$ given the allocation $A$ and $v_i(A(i))$ is its value.
Thus, for each $\{ i, j\} \subseteq N$, $A(i) \cap A(j) = \emptyset$ and any house is allocated to at most one agent. The set of all such allocations is denoted by $\mathcal{A}$.\footnote{Note that in its graphical representation, this model differs from the classical bipartite matching problem as it allows for allocations along zero edges (non-edges) in a graph.
}

\paragraph{Fairness.}
%EF, min envy, min total envy.
% Given a house allocation instance, our goal is to match agents to houses in a fair manner. We use the notion of `envy' and its variants to measure the fairness of a given allocation. 

Given an allocation $A$, we say that agent $i$ \emph{envies} $j$ if $v_{i}(A(j)) > v_i (A(i))$. The \textit{amount (magnitude)} of this pairwise envy is captured by $\envy_{i,j} (A) \coloneqq \max\{v_i (A(j)) - v_i (A(i)), 0\}$.
Given an allocation $A$, the \textit{total (aggregate) envy} of an agent $i$ towards other agents is denoted by $\envy_{i}(A) = \sum_{j\in N} \envy_{i,j}(A)$. 
%Moreover, the maximum amount of envy an agent $i$ experiences towards any other agent, termed as \textit{maximum envy} of agent $i$ in $A$, is $ \max_{j\in N} \envy_{i,j} (A)$.

% Given an allocation $A$, we denote by $\envy(A)$ the number of envious agents, that is, the size of the set $\{ i\in N: \envy_{i,j} (A) > 0,~\ \text{for some}\ j\in N\}$.
% Moreover, the {total envy} of an allocation $A$ is the aggregate amount of envy experienced by all agents, i.e. $\tenvy (A) \coloneqq \sum_{i\in N}\envy_{i}(A)$.

An allocation $A$ is \textit{envy-free} (\EF{}) if and only if for every pair of agents $i,j\in N$ we have $v_{i}(A(i)) \geq v_i (A(j))$, that is, $\envy_{i,j} (A) = 0$.
Since envy-free allocations are not guaranteed to exist, we consider other plausible approximations to measure the `degree of envy' \cite{chevaleyre2006issues}. %in a society of agents.
%
% min envy:
% min total envy:
% minimax envy (MME)
% minimax total envy (MMTE)

\paragraph{Degrees of envy.}
An allocation is \textbf{min \#envy} if it minimizes the number of envious agents, i.e.
    $
    \min_{A\in \mathcal{A}} \#\envy(A)
    $,
    where $\#\envy(A)$ is the number of envious agents, i.e., the size of the set $\{ i\in N: \envy_{i,j} (A) > 0, \text{for some}\ j\in N\}$.
\begin{sloppypar}
    An allocation is \textbf{min total envy} if it minimizes the total envy of all agents, i.e. 
    $
    \min_{A\in \mathcal{A}} \tenvy(A)
    $, 
    where $\tenvy (A)$ is the {total amount of envy} of allocation $A$ experienced by all agents, i.e., $\tenvy (A) \coloneqq \sum_{i\in N}\envy_{i}(A)$.
\end{sloppypar}
    An allocation is \textbf{minimax total envy} if it minimizes the maximum aggregate amount of envy experienced by an agent, i.e.  
    $
    \min_{A\in \mathcal{A}} \max_{i\in N} \envy_{i} (A)
    $.

% \begin{itemize}
        
%     \item \textbf{min \#envy}: an allocation that minimizes the number of envious agents is defined as
%     $$
%     \min_{A\in \mathcal{A}} \#\envy(A),
%     $$
%     where $\#\envy(A)$ is the number of envious agents, i.e., the size of the set $\{ i\in N: \envy_{i,j} (A) > 0, \text{for some}\ j\in N\}$.
  
%     \item \textbf{min total envy}: an allocation that minimizes the total envy of all agents is defined as 
%     $$
%     \min_{A\in \mathcal{A}} \tenvy(A),
%     $$

%     where $\tenvy (A)$ is the {total amount of envy} of allocation $A$ experienced by all agents, i.e., $\tenvy (A) \coloneqq \sum_{i\in N}\envy_{i}(A)$.

%   %  \item \textbf{minimax envy}: an allocation that minimizes the maximum amount of envy experienced by an agent is defined as 
%    % $$
%   %  \min_{A\in \mathcal{A}} \max_{i,j\in N} \envy_{i,j} (A).
%   %  $$

%     \item \textbf{minimax total envy}: an allocation that minimizes the maximum aggregate amount of envy experienced by an agent is defined as 
%     $$
%     \min_{A\in \mathcal{A}} \max_{i\in N} \envy_{i} (A).
%     $$
% \end{itemize}

% We seek to find allocations that minimize the aforementioned measures of envy.
In the above measures, an allocation is selected from the set of all feasible allocations $\mathcal{A}$. In the next section we discuss the reason behind restricting the set $\mathcal{A}$ to subsets that satisfy some measures of economic efficiency.

\paragraph{\textbf{Social Welfare.}}
% complete, max welfare, fixed size.
Without any measures of social welfare, any empty allocation is vacuously envy-free, and consequently satisfies all four measures of fairness. 
Hence, we consider three notions of social welfare that measure the economic efficiency of allocations based on their size (number of assigned agents), utilitarian, or egalitarian welfare.

% We consider three measures of economic efficiency in the society of agents based on the size of the allocated agents, the utilitarian welfare, or the egalitarian welfare.

The \textbf{size} $|A|$ of an allocation, $A$, is simply the number of agents that are assigned to a house.\footnote{We intentionally use the term `size' instead of `cardinality' to avoid confusion with matchings that \textit{only} allow selection of positively valued (aka. `liked') edges of a graph.}
An allocation is \textbf{complete} if it either assigns a house to every agent ($N$-saturating) when $m\geq n$, or assigns every house to an agent ($H$-saturating) when $m < n$.
%Hence, allocation $A$ is \textit{maximum size} if $|A| = \min \{m, n\}$.
Note that completeness is a weak efficiency requirement that does not take agents' valuations into account. 

% \HH{maybe move previous results on complete matching here.}

The \textbf{utilitarian welfare} of an allocation $A$ is the sum of the values of individual agents, i.e. $\sw(A) \coloneqq \sum_{i\in N} v_i (A)$. 
A maximum utilitarian welfare allocation is the one that maximizes the sum of the values, and can be found efficiently by computing a maximum-weight bipartite matching in the induced graph (a bipartite graph on $N\cup H$ where edge weights are given by the valuations $V$).  %\HH{double check!}

% via the Hungarian algorithm.

\begin{sloppypar}
 The \textbf{egalitarian welfare} of an allocation $A$ is the value of the worst off agent among all agents in $N$, that is, $\egwel(A) \coloneqq \min_{i\in N} v_i(A)$.
A $k$-egalitarian welfare is the value of the worst off agent in a subset $S\subseteq N$ of agents of size $k = |S|$ such that $\egwel(A) \coloneqq \min_{i\in S} v_i(A)$.
\end{sloppypar}
If it is possible to achieve a positive (non-zero) egalitarian welfare for all agents, any allocation that maximizes the egalitarian welfare (there could be multiple) can be selected. In the special case where every feasible $N$-saturating allocation (allocations of size $n$) has an egalitarian welfare of zero, we look for the largest subset of agents $S\subseteq N$ that can simultaneously receive a positive value, and select an allocation that maximizes the $k$-egalitarian welfare among these agents.
In \cref{ex:tradeoff}, a maximum egalitarian welfare allocation has welfare one and size two since any larger subset of agents will result in a egalitarian welfare of zero.

% \HH{Example of $k$ egalitarian.}

% Similarly, a \textit{maximum welfare} allocation is the one that maximizes the utilitarian welfare.

% Analogous to egalitarian welfare, we define \emph{max envy} (and max aggregate envy) of an allocation $A$ to be the maximum envy (resp. max aggregate envy) of an agent for $A$.  An allocation is \emph{min-max envy} (and \emph{min-max aggregate envy}) allocation if it minimizes the maximum envy (resp. aggregate envy) of an agent.

\paragraph{\textbf{Fairness-Welfare Trade-offs.}}
Our main objective is to investigate the trade-offs between the four fairness measures and various notions of welfare (i.e. economic efficiency). 
Thus, for each fairness-welfare pair, we define computational problems in the following way: 
%Thus, we define various computational problems in the following way:
%
Given an instance of the house allocation problem, $\langle N, H, V \rangle$, find an allocation $A$ that minimizes unfairness as measured by $F$ within the set of all allocations that maximize an efficiency measure of $E$,
where $F$ is min envy, min total envy, minimax envy, or minimax total envy and $E$ is max size, max \sw, or max \egwel.

Some standard graph theoretic notations and algorithms that we use are provided in \cref{app:prelim} for easy reference.

%\HH{include an example to illustrate each of the fairness measures.}

\section{Maximum Size Envy-free Allocations} \label{sec:EF}

We start by considering envy-freeness as our main constraint. The goal is to find an envy-free allocation under various welfare measures. Clearly, an envy-free complete allocation may not always exist. Further, an empty allocation is always an envy-free allocation of size $0$. The first objective is to find an envy-free allocation of maximum size. In other words, among the set of all envy-free allocations we find an allocation that maximizes the number of assigned agents (or houses when $m \leq n$). 
%
%
% \HH{potentially remove this algorithm; this may changes in \cref{sec:USWminimaxEnvy} and \cref{sec:EgalminimaxEnvy} as well as the Egalitarian algorithm.}
% \M{I've removed those sections and am rewriting this as the EF section}
%
In their paper \cite{GSV19envy} describe an efficient polynomial time algorithm to find an envy-free allocation in ordinal instances. However, they restrict themselves to complete envy-free allocations where all agents are assigned some house. Their algorithm returns `empty' when the number of available houses falls below the number of agents. 

Our algorithm relaxes this constraint and returns a maximum size envy-free allocation which need not be complete. 

%While this algorithm is of independent interest, we show that this technique can be effectively used to find welfare maximizing \EF{} allocations (utilitarian and egalitarian), maximum size \EF{} allocations, as well as those that minimize the maximum envy. \M{What was this about?}\SR{can be removed}

\paragraph{\textbf{Algorithm description.}} 
\cref{alg:algorithmenvyt} creates a bipartite graph where each agent is only adjacent to houses that are its most preferred. In other words, if there is an edge between an agent $i$ and a house $h$, then $v_i(h)$ is positive and maximum among all houses remaining in the graph. The algorithm proceeds by removing all houses in any inclusion-minimal Hall violators. 
This process repeats by updating the highest valued house among the remaining houses for each agent, and adding, or retaining, the corresponding edges.
After either all houses are considered or no Hall violator is found, the algorithm returns a \emph{maximum size allocation}  which is produced by union of a maximum size matching $M$ in the induced graph $G$ and a maximum size matching in the complement graph $\overline{G} - M$.

The next lemma gives a natural invariant for the algorithm.

\begin{algorithm}[t]
        \caption{A maximum size envy-free allocation}
    \label{alg:algorithmenvyt}
    \begin{algorithmic}[1] \small
    \REQUIRE A house allocation instance $\langle N, H, V \rangle$
    \ENSURE A maximum size envy-free allocation %s.t. no agent is envious towards another agent
    \FOR{each agent $i\in N$}
        \STATE Let $h_i^{\max} \in \argmax_{h\in H} v_{i}(h)$
    \ENDFOR
    \STATE Define $E_{\topp} = \{(i,h) | i \in N, h\in H, v_i(h) = v_i(h_i^{\max}), v_i(h) >0\}$ 
    % \STATE Define $E_{th} = \{(i,h) | i \in N, h\in H, v_i(h) +t \geq v_i(h^{\max}), v_i(h) >0\}$.
    \STATE Create a bipartite graphs $G = (N \cup H, E_{\topp})$ %and $G' = (N \cup H, E_{th})$.
    \IF{there exists a Hall violator $(N',H')$ in $G$}
        % \STATE Let $(N',H')$ be a Hall violator in $G$
        \STATE Delete $H'$ from $H$;\ \textit{removing houses that cause envy.}
        % \FOR{each agent $i\in N$}
        %     \STATE Let $h_i^{\max} \in \argmax_{h\in H} v_{i}(h)$            
        % \ENDFOR
        % \STATE Update $E_{\topp} = \{(i,h) | i \in N, h\in H, v_i(h) = v_i(h_i^{\max})\}$ 
        % \STATE Update the graph $G = (N \cup H, E_{\topp})$
    \ELSE \RETURN Allocation $A$ = union of maximum size matching $M$ in $G$ and a maximum size matching in $\overline{G} - M$.
    \ENDIF
    \STATE Go to line 1
    \end{algorithmic}
\end{algorithm}

\begin{restatable}{lemma}{lemhalldeletion}\label{lem:halldeletion}\label{lem:prop-algmwEF}
    Given a weighted instance, any house removed by \cref{alg:algorithmenvyt} cannot be a part of any envy-free allocation.
\end{restatable}

\begin{restatable}{theorem}{thmalgorithmenvyt}\label{thm:algorithmenvyt}
    Given a weighted instance, \cref{alg:algorithmenvyt} returns an envy-free maximum size allocation. %$A$ such that for every pair of agents $i,j\in N$, $\envy_{i,j}(A) = 0$.
\end{restatable}
\begin{proof}
First, note that \cref{alg:algorithmenvyt} runs in polynomial time because every component of the algorithm including finding a inclusion-minimal Hall violator~\cite{GSV19envy,aigner2022envy} and computing a maximum size bipartite matching runs in time polynomial in $n$ and $m$.
Therefore, it suffices to prove that 
i) every house removed by the algorithm cannot be contained in any envy-free allocation, and
ii) a maximum size bipartite matching  returns a maximum size allocation among all envy-free allocations.

Statement (i) immediately follows from \cref{lem:halldeletion} and the fact that each agent's valuations for the houses unassigned in $M$ is zero. Statement (ii) follows from the observation that \cref{alg:algorithmenvyt} finds a maximum size matching in the induced graph $G$, where every agent only has edges to its most preferred house among the houses retained, no further edges can be added, and no Hall violators exist. 

 The unassigned houses and unassigned agents in $M$ are assigned in a maximum matching in $\overline{G} -M$. Thus, no agent or house that is already assigned in $M$ is reassigned. Moreover, since we find a maximum matching in $\overline{G} -M$, it assigns maximum number of agents to the houses they value zero.
Thus, the algorithm returns the maximum size envy-free allocation on the instance.
\end{proof}

% Note that \cite{GSV19envy} consider the problem of finding a envy-free solution in ordinal instances where each agent has a ranking %\HH{do they handle weak?}\SR{yes, corrected}
% over the set of houses and ties are permitted. However, their approach    returns `empty' when the number of available houses falls below the number of agents. \SR{We already said something similar before. Remove this?}

\section{Utilitarian Welfare}\label{sec:utilwel}
In this section, we show that fair allocations can be obtained efficiently by introducing a utilitarian welfare constraint and show the connection between complete allocations and welfare maximizing ones when restricting the number of houses.

We begin the section by designing an EF allocation with maximum welfare.
\begin{restatable}{proposition}{thmWeightedEFmaxWelsize}
\label{thm:WeightedEFmaxsize}\label{thm:WeightedEFmaxWelfare}\label{thm:WeightedEFmaxWelsize}
    Given a weighted instance, an envy-free allocation of maximum \sw{} can be computed in polynomial time. 
\end{restatable}

Note that while a maximum size envy-free allocation can be computed in polynomial time, finding a complete allocation that minimizes the number of envious agents is \NPH{} \cite{KMS21complexity} (see example given in \cref{fig:mainexample}).

\subsection{Minimum \#Envy}
% We start by showing that in contrast to the hardness of finding a complete allocation that minimizes the number of envious agents \cite{kamiyama2021envy}, 
We start by showing that finding a \sw{}-maximizing allocation that minimizes \#envy can be done in polynomial time.
The algorithm constructs a bipartite graph and uses minimum cost  perfect matching~\cite{ramshaw2012minimum}. The details of the construction and proofs are relegated to \cref{app:subsec:memw}.

\begin{restatable}{theorem}{thmWeightedminEmaxWelfare}
\label{thm:WeightedminEmaxWelfare}
    Given a weighted instance, a \memw{} allocation can be computed in polynomial time.
\end{restatable}

While finding a \mec{} allocation is shown to be \NPH{} even for binary instances \cite{KMS21complexity}, we establish a relation between \mec{} and \memw{} allocations when $m\leq n$ by extending a minimum envy maximum \sw{} allocation to a complete allocation.\footnote{For binary instances, \citet{MadathilMS23} independently showed that a \mec{} allocation (termed as ``optimal'' house allocation) can be computed in polynomial time when $m\leq n$.} 

\begin{restatable}{proposition}{propWelCompmlessn}\label{prop:wel-comp-m-less-n}
In a binary instance when $m \leq n$, given a \memw{} allocation $A$, a complete allocation $\hat{A}$ can be constructed in polynomial time such that
\begin{enumerate}[(i)]
    \item\label{it:wel-comp-m-less-n1} $A$ and $\hat{A}$ has equal \sw{} and \#envy, and
    \item\label{it:wel-comp-m-less-n2} $\hat{A}$ has minimum \#envy among all complete allocations.
\end{enumerate}
\end{restatable}

% In next algorithm we show that when the valuations are arbitrary positive values and $m\leq n$, we can still find an \mec allocation in polynomial time by completing a maximum matching in the top preference graph.

 Note that \cref{prop:wel-comp-m-less-n} does not hold for weighted instances as we illustrate by \cref{ex:mecWtm-lessn} in \cref{app:subsec:memw}. %The next example illustrates this point.
Nevertheless, we show that a \mec allocation can be computed in polynomial time for weighted instances when $m \leq n$. The proof is relegated to \cref{app:utilwel}.
%using the observation that in any complete allocation an agent is not envious if and only if it is assigned to one of its most preferred houses.

\begin{remark}
First, given a maximum size envy-free allocation we cannot `append' it to achieve a \mec{} allocation even for binary instances. In \cref{ex:tradeoff}, the allocation indicated by orange lines cannot be simply completed to reach the \mec{} allocation (shown in green).
Second, when $m > n$ a \memw{} allocation may leave more agents envious compared to a \mec{} allocation. In \cref{ex:tradeoff}, the \memw{} allocation indicated by blue leaves two agents envious while the \mec{} allocation (shown in green) only leaves one agent envious. 
\end{remark}

\begin{restatable}{theorem}{thmWeightedminEnvycomplete}
\label{thm:WeightedminEnvycomplete}
    Given a weighted instance, when $m \leq n$, a \mec allocation can be computed in polynomial time. %$O(m\sqrt{n})$.\SR{change it to poly time instead of $m\sqrt{n}$?}
\end{restatable}

   Our approach to find a \memw{} allocation heavily relies on the set of \sw{} maximizing allocations. Next, we discuss two  observations about the necessity of utilizing \sw{} maximizing allocations and restriction on the number of houses.

\subsection{Minimum Total Envy}

When focusing on total envy of agents, under binary valuations, the total envy can be seen as the number of distinct envy relations between all pair of agents. Whereas in weighted instances, individual values for each assigned houses (and not just pairwise relations) play an important role in computing the total envy.

\begin{example} \label{ex:maxwelfareDifferentEnvy}
    Consider the weighted instance given in \cref{fig:maxwelfareDifferentEnvy}. There are two allocations that maximize the \sw{} with a total welfare of $15$. 
    One allocation (shown in red) leaves three agents envious (namely, $a_2$, $a_3$, and $a_4$) with a total envy of $5$; while another allocation (shown in green) only contains one envious agent ($a_1$) and still generates the total envy of $5$.    
    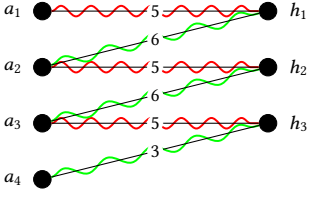
\begin{figure}[t]
        \centering
        \begin{minipage}{0.45\linewidth}
        {\footnotesize
        \begin{tikzpicture}[every node/.style={draw,circle}, fsnode/.style={fill=black}, ssnode/.style={fill=black}, vnode/.style={draw=white, fill=white, inner sep=0pt,outer sep=0pt, midway}]
        
        % the agents
        \begin{scope}[start chain=going below,node distance=5mm]
        \foreach \i in {$a_1$, $a_2$, $a_3$, $a_4$}
          \node[fsnode,on chain] (f\i) [label=left: \i] {};
        \end{scope}
        
        % the houses
        \begin{scope}[xshift=3cm,yshift=0cm,start chain=going below,node distance=5mm]
        \foreach \i in {$h_1$, $h_2$, $h_3$}
          \node[ssnode,on chain] (s\i) [label=right: \i] {};
        \end{scope}
        
        % the edges
        \tikzset{decoration={snake,amplitude=.7mm,segment length=4mm,
                       post length=0mm,pre length=0mm}}

        \draw[decorate, red, line width=0.75] (f$a_1$) -- (s$h_1$);
        \draw[decorate, red, line width=0.75] (f$a_2$) -- (s$h_2$);
        \draw[decorate, red, line width=0.75] (f$a_3$) -- (s$h_3$);

        \tikzset{decoration={snake,amplitude=.7mm,segment length=6mm,
                      post length=0mm,pre length=0mm}}
        \draw[decorate, green, line width=0.75] (f$a_2$) -- (s$h_1$);
        \draw[decorate, green, line width=0.75] (f$a_3$) -- (s$h_2$);
        \draw[decorate, green, line width=0.75] (f$a_4$) -- (s$h_3$);

        \draw (f$a_1$) -- (s$h_1$) node[vnode] {$5$};
        \draw (f$a_2$) -- (s$h_2$) node[vnode] {$5$};
        \draw (f$a_2$) -- (s$h_1$) node[vnode] {$6$};
        \draw (f$a_3$) -- (s$h_2$) node[vnode] {$6$};
        \draw (f$a_3$) -- (s$h_3$) node[vnode] {$5$};
        \draw (f$a_4$) -- (s$h_3$) node[vnode] {$3$};
        
       \end{tikzpicture}}
        \end{minipage}\hfill
        \begin{minipage}{0.45\linewidth}  
   
        \caption{\small{An example depicting the difference between min \#envy and min total envy in \sw{}-maximizing allocations. Red has three envious agents, namely, $a_2,a_3,a_4$ and green has one envious agent $a_1$ but both has total envy five.}}
       \label{fig:maxwelfareDifferentEnvy}
        \end{minipage}
    %
    % \vspace{-10pt}
\end{figure}
\end{example}

This example illustrates that the proposed algorithms for finding \memw{} cannot be readily used for finding \mtemw{}. In \cref{app:utilwel}, we present an example (\cref{ex:envy_totalenvy_binary}) which shows that this challenge persists even for binary instances. 
Nonetheless, we show that one can achieve a \mtemw{} allocation in polynomial time by constructing a bipartite graph, similar to the algorithm for \memw{}, with a carefully crafted cost function that encode envy and \sw{} as cost, and utilizing algorithms for finding a minimum cost perfect matching. %as described in below and in \cref{alg:minTotalEnvyMaxWelfare}. 

\paragraph{\textbf{Algorithm description}.}
The algorithm (\cref{alg:minTotalEnvyMaxWelfare} in \cref{app:subsec:mtew}) proceeds by constructing a bipartite graph $G = (N \cup H', E)$ where the set $H'$ is constructed by adding a set of $n$ dummy houses to the set of houses $H$. That is, $H' = H \cup \{h^i \mid i \in N\}$. Given an agent $i\in N$ and $h\in H'$, we have $(i, h) \in E$ if and only if $h\in H$ and $v_i(h) >0$, or $h \in H' \setminus H$. 
 We construct a cost function $c$ on edges of $G$. Before we define the cost function we scale the valuations $V$ such that for each agent $i$ and house $h$, if we have that $v_i(h) >0$, then $v_i(h) \geq 1$. Now we define two components of the cost function $c$, namely, envy component $c_{envy}$ and \sw{} component $c_{sw}$ such that $c = c_{envy} + c_{sw}$.
  Let $H^{\max}_i$ be the set of most preferred houses in $H$ for agent $i$, i.e., $H^{\max}_i = \argmax_{h \in H} v_i(h)$. For ease of exposition we assume $v_i(h) = 0$ for each agent $i$ and a dummy house $h\in H'\setminus H$. For an edge $(i, h) \in E$, if $ h \in H^{\max}_i$, then define $c_{envy}(i, h)  = 0$; otherwise $c_{envy}(i, h)= \sum_{\bar{h} \in H}\max\{v_i(\bar{h}) - v_i(h), 0\}$.
%  \begin{equation*}
%      c_{envy}(i, h)  = \begin{cases}
%          0,& \textit{ ~~if } h \in H^{\max}_i\\
%        % \sum_{\bar{h} \in H} v_i(\bar{h}), &\textit{ ~~if } h \in H'\setminus H\\
%        \sum_{\bar{h} \in H}\max\{v_i(\bar{h}) - v_i(h), 0\}& \textit{~~otherwise}.
%      \end{cases}
% \end{equation*} 
% \HH{to save some space, we can write these cases inline rather than equations.}
We denote $\sum_{i\in N}\sum_{\bar{h} \in H} v_i(\bar{h}) +1$ by $L$. Furthermore, we define the \sw{} component of the cost $c$ for an edge $(i,h) \in E$ as   $c_{sw}(i, h)  =-v_i(h)\cdot L$.
% \begin{equation*}
%      c_{sw}(i, h)  = \begin{cases}
%          -v_i(h)\cdot L& \textit{ ~~if } h \in H\\
%      0 &\textit{ ~~if } h \in H'\setminus H.
%      \end{cases}
% \end{equation*} 
 Finally, we return a \mcmw{} matching in $G$.

\begin{restatable}{theorem}{thmminTotalEnvyMaxWelfare}
\label{thm:minTotalEnvyMaxWelfare}
    Given a weighted instance, a \mtemw{} allocation can be computed in polynomial time.
\end{restatable}
\begin{proof}[Proof Sketch]
    Suppose the algorithm (\cref{alg:minTotalEnvyMaxWelfare}) returns allocation $A$. Then $A$ corresponds to a \mcmw{}  in the constructed graph $G$. We show that by minimizing the cost, we maximize the welfare and minimize the total envy. Observe that the cost of each pair $(i,h) \in A$ has two components, namely, \sw{} component $-v_i(h)\cdot L$ and a envy component. To complete the proof we show that (i) since $L$ is large, a minimum cost matching in $G$ maximizes \sw{}; (ii) the envy component of the cost of a perfect matching correctly computes the total envy of each agent. It follows from the fact that in a max \sw{} allocation every house valued higher than house $h$ by agent $i$ must be allocated for each $(i,h) \in A$.
\end{proof}

We show a result analogous to \Cref{prop:wel-comp-m-less-n} holds for min total envy even for weighted instances.

\begin{restatable}{proposition}{propWelCompmlessTE}\label{prop:TE-wel-comp-m-less-n}
 % \SR{@Hadi and  Medha, could you please check this proof?}\M{It looks good to me}
% Given a \mtemw allocation, one can construct a \mtec allocation in polynomial time when $m \leq n$.
Given a weighted instance with $m\leq n$, let $A^*$ be a \mtemw allocation. Then a complete allocation $A$ can be constructed in polynomial time such that
\begin{enumerate}
    \item $A^*$ and $A$ has equal \sw{} and envy, and
    \item $A$ is a \mtec allocation.
\end{enumerate}
\end{restatable}

%\HH{explain the intuition behind the next result; and its difference from max welfare.}
%\SR{We don't need the next Theorem separately, we can merge the statement with \cref{prop:TE-wel-comp-m-less-n} }
%\HH{OK, please merge.}\SR{changed it to corollary}

\cref{cor:minTotalEnvyComplete-mleqn} follows immediately from \cref{prop:TE-wel-comp-m-less-n}.

\begin{restatable}{corollary}{thmminTotalEnvyComplete}
\label{cor:minTotalEnvyComplete-mleqn}
    Given a weighted instance, a \mtec allocation can be computed in polynomial time when $m \leq n$.
    % \footnote{\cite{MadathilMS23} independently showed a similar result but only for binary valuations under the name ``utilitarian house allocation''.\SR{We state it in the table. Do we need to say it here as well?} \HH{no let's remove to save space.}} 
\end{restatable}
  
% \begin{proof}
%     \HH{TBA}
% \end{proof}

%\HH{update this to use \cref{ex:tradeoff}. \SR{\cref{ex:tradeoff} has $m\geq n$}}
Observe that even when $m \leq n$, there may be several complete matchings with different total envy.
In \cref{fig:completematchings}, both matchings are complete because they assign all the houses, however, the allocation indicated by red yields a higher total envy (two by $a_2$) than the green one (one by $a_3$). 

    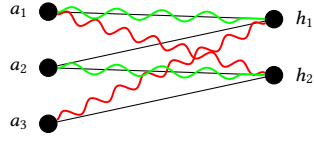
\begin{figure}[t] 
    \centering
    \begin{minipage}{0.45\linewidth}
        \centering
        {\footnotesize
        \begin{tikzpicture}[every node/.style={draw,circle}, fsnode/.style={fill=black}, ssnode/.style={fill=black}]
        
        % the agents
        \begin{scope}[start chain=going below,node distance=5mm]
        \foreach \i in {$a_1$, $a_2$, $a_3$}
          \node[fsnode,on chain] (f\i) [label=left: \i] {};
        \end{scope}
        
        % the houses
        \begin{scope}[xshift=3cm,yshift=-0.1cm,start chain=going below,node distance=5mm]
        \foreach \i in {$h_1$, $h_2$}
          \node[ssnode,on chain] (s\i) [label=right: \i] {};
        \end{scope}
        
        % the edges
        \tikzset{decoration={snake,amplitude=.7mm,segment length=4mm,
                       post length=0mm,pre length=0mm}}
        \draw (f$a_1$) -- (s$h_1$);
  
        \draw (f$a_2$) -- (s$h_1$);
        \draw (f$a_2$) -- (s$h_2$);
        % \draw (f$a_2$) -- (s$h_1$);
        \draw (f$a_3$) -- (s$h_2$);
    
        \draw[decorate, red, line width=0.75] (f$a_1$) -- (s$h_2$);
        \draw[decorate, red, line width=0.75] (f$a_3$) -- (s$h_1$);

        \tikzset{decoration={snake,amplitude=.7mm,segment length=6mm,
                      post length=0mm,pre length=0mm}}
        \draw[decorate, green, line width=0.75] (f$a_1$) -- (s$h_1$);
        \draw[decorate, green, line width=0.75] (f$a_2$) -- (s$h_2$);
        
       \end{tikzpicture}}
       
       \end{minipage}\hfill
       \begin{minipage}{0.46\linewidth}

        \caption{\small{Two complete allocations with different total envy. Total envy of the red allocation is two due to agent $a_2$, the same for green is one due to $a_3$.}}
        \label{fig:completematchings}
    \end{minipage}
% \vspace{-16pt}  
%
\end{figure}
  
\section{Egalitarian Welfare}\label{sec:egal}

When the efficiency measure is maximizing the utilitarian welfare, any maximum-weight matching on the induced bipartite graph returns the required allocation in polynomial time.
However, the problem of finding an allocation that maximizes the egalitarian welfare has received less attention in the house allocation setting.

% In general, the problem of maximizing the egalitarian welfare (worst-off agent) is shown to be \NPH{}.

While in fair division finding an allocation that maximizes the egalitarian welfare is \NPH{}\footnote{When agents can receive multiple items, an egalitarian allocation always exists but computing a  egalitarian allocation is \NPH{} \cite{bouveret2016characterizing}.
},
we show that in the house allocation setting wherein agents are restricted to receive at most one house, an egalitarian solution can be found in polynomial time.

Note that in binary instances, finding an egalitarian allocation is equivalent to finding an envy-free allocation of maximum size (\cref{prop:EFmaxSize}) since in every allocation the egalitarian welfare is either zero or one.
When it comes to weighted instances, however, the goal is to maximize the number of agents who receive a positive value and, conditioned on that, maximize the value of the worst-off agent.
We use this intuition to search for an allocation that maximizes the number of agents that receive a positively valued house.

\begin{algorithm}[t]\small
        \caption{Finding an allocation of max \egwel{}}\label{alg:algorithm-egalk}
    \begin{algorithmic}[1]
    \REQUIRE A house allocation instance $\langle N, H, V \rangle$.
    \ENSURE An allocation with maximum egalitarian welfare.
    \FOR{$k = n$ to $1$}
        % \STATE Create a bipartite graph $G = (N \cup H, E)$ s.t. there is exists an edge between each agent $i\in N$ and each house $h\in H$ if $v_{i}(h) > 0$ 
        \STATE Let $v_{\max} \coloneqq \max_{i,h}{v_{i}(h)}$ and $v_{\min} \coloneqq \min_{i,h}{v_{i}(h)}$
        
        $\triangleright$ \textsc{\footnotesize Looping  all unique values of houses in $H$}
        
        \FOR{$\beta = v_{\max}$ to $v_{\min}$}
             \STATE Create a bipartite graph $G_{\beta} = (N \cup H, E)$ s.t. there is exists an edge between each agent $i\in N$ and each house $h\in H$ if $v_{i}(h) \geq \beta$ 
            % \STATE Remove every edge $(i',h')$ from $G$ if $v_{i'}(h') < \beta$
            \STATE Let $A \coloneqq$ a maximum size matching of $G_{\beta}$
            \IF{there exist $k$ allocated agents in $A$}
                \RETURN Allocation $A$
            \ENDIF   
        \ENDFOR
    \ENDFOR
    % \RETURN $\emptyset$
    \end{algorithmic}
\end{algorithm}

\paragraph{\textbf{Algorithm description.}}
The algorithm (\cref{alg:algorithm-egalk}) begins by considering the maximum number of agents $k=n$ who can potentially receive a positively-valued house. 
Consider the set of all agent valuations in decreasing order. 
For each such value $\beta$, create a bipartite graph $G_{\beta}$ where there is an edge between any agent-house $(i,h)$ pair if $v_{i}(h) \geq \beta$.
Now, find a maximum-size matching $M$ on $G_{\beta}$.
If the size of the matching is $k$, the algorithm returns $M$ as the required allocation. 
Otherwise, it repeats this process by decreasing the size to $k = k-1$. %The algorithm (\cref{alg:algorithm-egalk}) is in \cref{app:egal}.

\begin{restatable}{theorem}{thmEgalWelfare}
\label{thm:EgalWelfare}
    Given a weighted instance, an egalitarian welfare maximizing allocation can be found in polynomial time.
\end{restatable}

\begin{proof}
Suppose that \cref{alg:algorithm-egalk} returns an allocation $A$ of size $k$ where every allocated agent receives positive value of at least $\beta$. To prove the theorem we need to show that (i) $k$ is the largest number of agents that can simultaneously receive positive value and (ii) $\egwel(A)$ is maximum among all allocations of size at least $k$. 
Note that it is sufficient to prove this for size exactly $k$ since we cannot increase \egwel{} by increasing the size of $A$.
Since $A$ is a maximum size matching in $G_{\beta}$, from definition of $G_{\beta}$, it holds that we cannot allocate more than $k$ agents to the houses they value at least $\beta$. Furthermore, the size $k$ decreases from $n$, and the algorithm returns the first $k$-sized allocation where each assigned agent receives positive value. %Thus, at most $k$ agents can receive positive value in any allocation.
Thus $A$ is the largest possible allocation where each assigned agent gets some positive value since we iterate over all positive values of $\beta$. Thus we show (i).
Moreover, since we start by setting $\beta$ to the highest possible value of \egwel{} and decrease step by step, there does not exist a $k$ size allocation for a higher value of $\beta$.   Therefore, for any $k$-sized allocation $\beta$ is the maximum \egwel{}. Thus we show (ii).
Note that the possible values of $\beta$ is bounded by the distinct values agents  have towards the houses. Then, there are $O(mn)$ values can be assumed by $\beta$. Thus the algorithm runs in polynomial time.
\end{proof}

Clearly, a maximum egalitarian allocation may not be unique. Thus, a natural question is whether we can find a fair allocation among all such allocations. 
We first show that analogous to its utilitarian counterpart (\cref{thm:WeightedEFmaxWelfare}), an envy-free allocation (if one exists) of maximum \egwel{} can be computed in polynomial time.

\paragraph{\textbf{Algorithm description}.}
Given an instance $I = \langle N, H, V \rangle$, \cref{alg:maxegalEF} first finds a max \egwel{} allocation with  welfare at least $\beta$ for $k$ agents using \cref{alg:algorithm-egalk}. Then we find an envy-free allocation, if exists, with egalitarian welfare at least $\beta$ for $k$ agents. %We construct the bipartite  graph $G_\beta$ on vertex set $N \cup H$ by adding an edge between an agent $i \in N$ and a house $h \in H$ if $v_i(h) \geq \beta$.  
We construct a reduced valuation $V'$ where $v'_{i}(h)$ is set to $v_i(h)$ if $v_i(h)\geq \beta$ and is zero otherwise for an agent $i \in N$ and house $h \in H$.
Then, we invoke \cref{alg:algorithmenvyt} as a subroutine to find a EF allocation $A$ in $\langle N, H, V' \rangle$. If $k$ agents receive value at least $\beta$ in $A$, then we return the allocation $A$; otherwise we return an empty allocation since no EF allocation of max \egwel{} exists.
%The algorithm (\cref{alg:maxegalEF}) and the proof of the following theorem are in \cref{app:egal}.
% If $A$ assigns $k$ agents a house that they value $\beta$, then we return $A$; else we assign the next smaller value to $\beta$ and repeat. If $\beta$ reaches zero, then we return an empty allocation.

\begin{algorithm}[t]\small
        \caption{Finding an EF allocation of maximum \egwel{}}\label{alg:maxegalEF}
    \begin{algorithmic}[1]
    \REQUIRE A house allocation instance $\langle N, H, V \rangle$.
    \ENSURE An EF allocation with maximum egalitarian welfare.
    \STATE Let $k$ and $\beta$ denote the size and \egwel{} of an allocation returned by \cref{alg:algorithm-egalk}.
    \STATE Create valuation $V'$ s.t $v'_{i}(h)= v_i(h)$ if $v_i(h)\!\geq \beta$; $v'_{i}(h)\!=\!0$ otherwise, for an agent $i \in N$ and house $h \in H$.
    \STATE Let $A$ be the allocation returned by \cref{alg:algorithmenvyt} given $\langle N, H, V' \rangle$.
    \IF{there exists $k$ allocated agents in $A$}
        \RETURN allocation $A$. %that is a maximum weight matching of $G'$. 
    \ENDIF
    \RETURN $\emptyset$
    \end{algorithmic}
\end{algorithm} 
We prove the correctness of \cref{alg:maxegalEF} in the next theorem.

\begin{restatable}{theorem}{thmEFegalwel}\label{thm:EFegalwel}
    Given a weighted instance, an envy-free allocation of maximum \egwel{} can be computed in polynomial time.    
\end{restatable}

\begin{proof}
We have the following property of the allocation $A$ returned by \cref{alg:maxegalEF}.
     Each agent that is assigned a house in $A$ receives value at least $\beta$.
The statement follows from the facts that each agent is assigned a house that it values positively by \cref{alg:algorithmenvyt} and each positive value in $V'$ is at least $\beta$. 
However, there may exist agents that is not assigned any house in $A$. Thus, if \cref{alg:algorithmenvyt} returns an allocation that does not assign $k$ agents, from the definition of maximum egalitarian welfare we conclude that there is no EF allocation of max \egwel{}.  The correctness of this step follows from \cref{thm:algorithmenvyt}. It shows that in $A$, maximum number of agents are assigned with positive value.  
%if there is an agent $i$ who is not assigned any house, then all the houses $i$ likes are deleted when we deleted the Hall violators. 
So an agent $i$ that is unassigned in the allocation $A$ cannot be assigned to a house it likes without generating envy. Therefore, there is no EF allocation that can match $i$ to a house that it values $\beta$ or more. 

Using \cref{thm:algorithmenvyt}, we have that allocation $A$ is EF. This completes the proof of correctness of \cref{alg:maxegalEF}.
Since  Algorithms~\ref{alg:algorithmenvyt} and \ref{alg:algorithm-egalk} runs in time polynomial in $n$ and $m$, \cref{alg:maxegalEF} runs in polynomial time.
\end{proof}

\begin{figure*}
    \centering
    \includegraphics[scale=0.38]{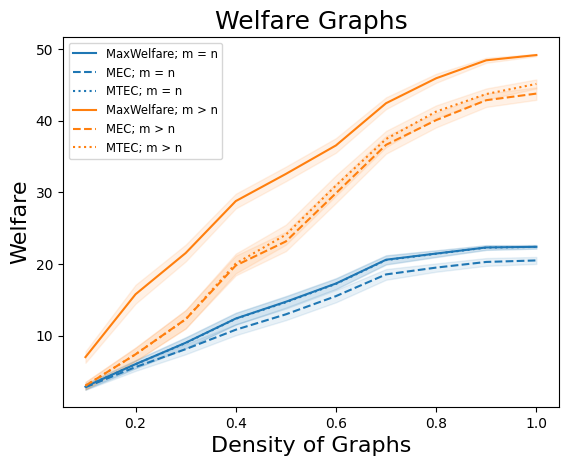}
    \includegraphics[scale=0.38]{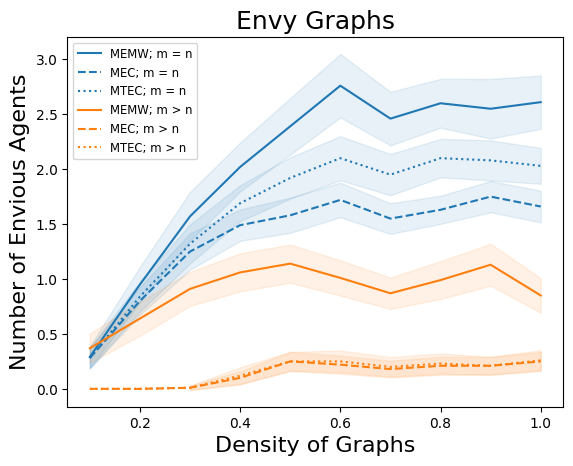}    
    \includegraphics[scale=0.38]{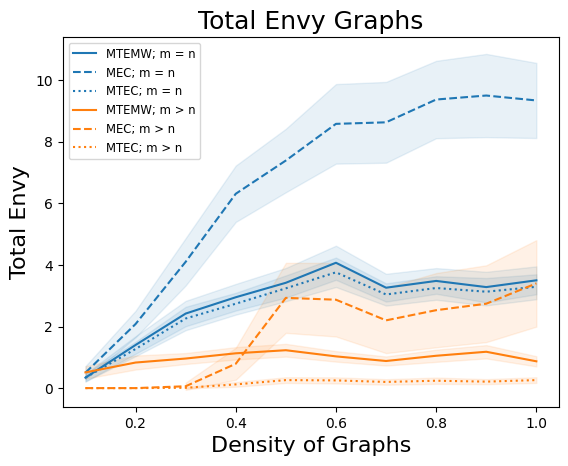} 
    \caption{The averages of \#envy (number of envious agents), total envy, and \sw{} over $100$ random trials for each graph density, given $5$ agents and $\{5, 10\}$ houses. Codes in the legend refer to MTEC: \mtec{}; MEC: \mec{}; MEMW: \memw{}; MTEMW: \mtemw{}.}
    \label{fig:truncborda}
\end{figure*}

\subsection{Minimum \#Envy}

We aim to find an \egwel{} welfare maximizing allocation that minimizes the number of envious agents. Under binary valuations, the \egwel{} is either zero or one. When the \egwel{} is one, we return a complete, envy-free allocation. For \egwel{} zero, an empty allocation is the optimal solution.
In contrast to the utilitarian welfare setting (\cref{thm:WeightedminEmaxWelfare}), finding a max \egwel{} allocation that is \minEnvy{} is intractable in a weighted instance. 
% \HH{should we also mention binary case?}

\begin{restatable}{theorem}{thmminenvykegal}\label{thm:minenvyk-egal}
    Given a weighted instance, finding a \minEnvy{} max \egwel{} allocation is \NPH{}.
\end{restatable}

\begin{proof}[proof sketch]
We prove this by showing a reduction from the problem of finding a \mec{} allocation that is known to be NP-complete~\cite{KMS21complexity}.
    Given an instance $I = \langle N, H, V\rangle$ of the minimum \#envy complete problem, %where the goal is to find a complete allocation with at most $z$ envious agents.  %For $k = n$ and envy $=z$,  
    we build an equivalent instance of the min \#envy max \egwel{} problem. For each agent $i \in N$ and house $h 
    \in H$, we create the valuation $v'_i(h)$ by adding  a positive small value $\beta$ to the valuation $v_i(h)$.
    Thus, all 
    max \egwel{} allocations assign a positively valued house to each agent under the new valuation. We show the hardness of the problems lies in minimizing \#envy under the completeness requirement.
    In \cref{app:egal:minenvy} we show the equivalence of the two instances to complete the proof.
\end{proof}
% \HH{is this problem \NPH{} for binary instances?}\SR{not in our reduction}
It is easy to check that the decision version of min \#envy max \egwel{} - where we check if there exists a max \egwel{} allocation with \#envy at most $t$ - is in NP. Thus the problem is \NPC{}.

\subsection{Minimum Total Envy}
While minimizing the total amount of envy experienced by the agents, restricting the search space to the maximum egalitarian welfare allocations does not result in any computational improvement. The problem remains computationally as hard as finding a \mtec{} allocation. 

\begin{restatable}{theorem}{thmredfromMTEC}\label{thm:redfromMTEC}
    Given a weighted instance, finding a \minTEnvy{} of max \egwel{} is as hard as finding a \mtec{} allocation.
\end{restatable}

The proof uses the same construction as in \cref{thm:minenvyk-egal}. We defer the details to \cref{app:egal:minTE}.
If in return the objective is to minimize the maximum total of envy  experienced by agents (minimax total envy) the problem becomes \NPH{}.
Note that similar to \cref{thm:minenvyk-egal}, in the decision version of the problem, one can check whether minimax total envy of the allocation is at most $t$, implying that the problem is \NPC{}.

\begin{restatable}{theorem}{thmmmtenvyegalwelfare}
    \label{thm:mmtenvy-egal-welfare}
    Given a weighted instance, finding a
    \minmaxTEnvy{} max \egwel{} allocation is \NPH{}. 
% \HH{how about the binary version?}
\end{restatable}

The detailed proof can be found in the \cref{app:egal:minimaxE}. In a nutshell, we provide a reduction from the Independent Set problem in cubic graph~\cite{GJ79}.
Note that even though our construction is similar to \cite{MadathilMS23}'s hardness reduction for finding a minmax total envy complete allocation\footnote{\citet{MadathilMS23} refer to this problem as ``egalitarian house allocation''.}, our reduction further ensures that every agent receives a positive value.

\section{Experiments}\label{sec:expt}
We experimentally investigate the welfare loss and fairness of the proposes algorithms on randomly generated bipartite graphs.
For a fixed number of agents, we varied the number of houses ($m = n$ and $m > n$) and considered both binary and weighted valuation functions $V$. We modelled preferences by iterating over the density of edges ($\lambda \in [0.1, 1.0]$), i.e. the probability that an edge exists is $\lambda$, in the corresponding bipartite graph. For each instance, defined by $(m, \lambda, V)$, we ran $100$ trials, on a randomly generated a graph $G$ that satisfied the $(m, \lambda, V)$ constraints. We compared the maximum \sw{} achieved by a welfare maximizing allocations, to the welfare achieved by an envy minimizing allocation allocations to understand the price of fairness. For the former allocations, we observed \memw{} and \mtemw{} allocations, and for the latter we considered \mec{} and \mtec{} allocations. Next, we compared the \#envy (and total envy) of \mec{} (resp. \mtec{}) with that of the \memw{} and \mtec{} (resp. \mtemw{} and \mec{}). The plots below show us the average value of these metrics over the $100$ trials for weighted valuations, and highlight the $95\%$-Confidence Interval.

\textbf{Observations.}
When there is an abundance of houses, the envy and total envy of all allocations decreases and the \sw{} increases. Similarly, as the graph grows denser (i.e. $\lambda > 0.4$), welfare increases and, under binary valuations, envy and total envy vanish. For weighted valuations too, we notice a slight decrease, but they still persist. As expected, the number of envious agents in a \mec{} allocation is least, followed by \mtec{} and \memw{}. The lower \sw{} of \mtec{} can be attributed to leaving highly valued houses unallocated. Notably, \mec{} has higher total envy than \mtemw{}, since it would prefer one highly envious agent to multiple slightly envious ones. Additional plots and discussions on experiments can be found in Appendix~\ref{expt-contd}.

\section{Concluding Remarks}
Our investigation on the tradeoffs between different efficiency and fairness concepts gives rise to several intriguing open questions. For example, the computational complexity of minimizing total envy remains unsolved. Moreover, one can ask if we can guarantee approximations of welfare to achieve EF or relaxations of EF; or whether the complexity of the problems change when considering strict ordinal preferences, Borda valuations, or pairwise preferences.

\section*{Acknowledgments}
Hadi Hosseini acknowledges support from NSF IIS grants \#2144413 and \#2107173.
We thank the anonymous reviewers for their helpful comments.

%%%%%%%%%%%%%%%%%%%%%%%%%%%%%%%%%%%%%%%%%%%%%%%%%%%%%%%%%%%%%%%%%%%%%%%%

%%% The next two lines define, first, the bibliography style to be 
%%% applied, and, second, the bibliography file to be used.
\clearpage

\bibliographystyle{named}
\bibliography{references}

%%%%%%%%%%%%%%%%%%%%%%%%%%%%%%%%%%%%%%%%%%%%

\clearpage

\appendix

\section*{Supplementary Material}

\section{Additional Related Work}\label{app:relatedwork}
Gan et al.~\cite{GSV19envy} %modeled the problem as a bipartite graph and 
described an algorithm to find an envy-free allocation, should one exist, that allocates a house to each agent when the preference of an agent is given as ranking over the houses. They also proved that an EF allocation exists with high probability if the number of houses exceeded the number of agents by a logarithmic factor. Building on this, Aigner-Horev and Segal-Halevi \cite{aigner2022envy}  developed an algorithm to find the maximum size EF allocation under binary valuations where agents are only assigned to houses they value positively. Further, for a slightly relaxed definition of envy under weighted valuations, they also found the maximum cardinality EF matching. %of minimum weight. 
%
%Kamiyama et al.~\cite{kamiyama2021envy} proved that minimizing the number of envious agents is \NPH{} under both general and binary valuations by providing a reduction from the maximum biclique problem and they allow agents to be matched to houses that they do not value. %Further, Kamiyama et al. also proved that finding an approximate solution to minimizing the number of envious agents is also \NPH{}.
%
Madathil et al.~\cite{MadathilMS23} consider complete allocations, allow assignment of zero-valued house to an agent. Under binary valuations, they study min \#envy, min total envy, and minimax total envy allocations and refer to them as optimal, utilitarian, and egalitarian house allocation problems, respectively.\footnote{Note that they are different from utilitarian or egalitarian welfare.} They show that %minimizing any agent's maximum envy towards another agent, i.e. for all agents $a$ and $a'$ in allocation $A$, minimizing $\max_{a, a'}(v_a(A(a'))-v_a(A(a)))$
minimax total envy complete allocations can be found in polynomial time under restrictions. %, e.g., when $m=n$. %~\cite[Proposition 2]{MadathilMS23}
% or agent valuations have extremal intervals, or every agent approves exactly one house. %~\cite[Lemma 36]{MadathilMS23}.
 They show it is \NPH{} in general using a reduction from Independent Set~\cite[Lemma 38]{MadathilMS23} similar to us (\Cref{thm:mmtenvy-egal-welfare}). %- even when the allowed value of the maximum envy is $1$. They also show that the problem is \NPH{} when the agents' preferences are ordinal and the target value of the maximum envy is $1$~\cite[Lemma 39]{MadathilMS23}. 
%We use a slight modification of this reduction to show that minimizing the maximum total envy an agent feels towards another agent under an allocation is \NPH{} even when the egalitarian welfare is strictly better than zero. The idea is to change
%We modify the zero and one valuations in their reduction  to $k$ and $k+1$, respectively, to prove hardness for minimax total envy max egalitarian welfare (\Cref{thm:mmtenvy-egal-welfare}).

Kamiyama~\cite{kamiyama2021envy} considered the problem of finding envy-free allocations for pairwise preferences and showed it is \NPH{} even with some restricted preferences, polynomial time with more restrictions, and \WH{} parameterized by the number of agents.
Belahcene et al.~\cite{belahcene2021combining} look at the house allocation problem under the guise of project allocations and evaluate a relaxed notion of envy-freeness i.e. rEF where agents are only envious of houses given to another agent if they rank the house higher than the other agent. Hosseini et al.~\cite{hosseini2023graphical} discuss the notion of aggregate envy, where they sum every agent's pairwise envy with every other agent to create an envy measure they minimize. They prove that even under restricted valuations, i.e. identical and evenly spaced agent valuations, minimizing the amount of aggregate envy is \NPH{} by a reduction from the linear arrangements problem. Hosseini et al.~\cite{hosseini2023tight} define the graphical housing allocation problem as generalization of the minimum linear arrangement and house allocation problem, and characterize the approximability of housing allocation on graphs with different structures.

\section{Additional Material from Section \ref{sec:model}}\label{app:prelim}

\subsection{Graphical Representations and Techniques}
%\HH{maximum matching, min cost max matching, Hall violator, what else?}
For completeness, here we define some of the standard definitions in graph theory that we used.
Given a bipartite graph $G= (A\cup B, E)$, a matching $M$ is a pairwise vertex disjoint subset of edges $E$. We denote the complement graph of a graph $G$ as $\overline{G}$. We use the notation $G - M$ to denote the reduced graph obtained by deleting the vertices matched in $M$ from $G$.
\begin{definition}[Hall Violator]
    A Hall set is a subset $A' \subseteq A$ such that $|N(A')|<|A'|$ where $N(A')$ denotes the neighbors of vertices in the set $A$, i.e., $N(A') = \{v \in B \mid u \in A, (uv) \in E\}$. We call the set $N(A')$ as a Hall violator. 
\end{definition}
A minimal Hall violator can be computed in polynomial time~\cite{GSV19envy,aigner2022envy}.

A matching $M$ is maximal when it is not contained in any other matching.
A maximum matching in a graph $G$ is a matching containing maximum number of edges in $G$. 
\begin{definition}[Maximum Size allocation in $G$]\label{def:maxbiaprtite}
    A maximum size bipartite matching in $G= (A\cup B, E)$ is the largest matching between the sets $A$ and $B$ formed by the union of a maximum size matching $M$ in $G$ and a maximum size matching in $\overline{G} - M$.
\end{definition}

Note that any maximum size allocation in a graph $G = (N \cup H, V)$ can be found in polynomial time, since it is the union of two maximum size matchings which can each be found in polynomial time.

Given a  bipartite graph $G= (A \cup B, E)$ with a cost function $c:E\mapsto\mathcal{R}$, the \mcmw problem (also known as the assignment problem) is to find a  matching $M$ that matches all vertices of the smaller size and minimizes the cost $\sum_{e\in M} c(e)$. Then, a \mcmw matching can be found in strongly polynomial time using Hungarian method~\cite{ramshaw2012minimum}.

% Corollary 12.2b]{schrijver2003combinatorial} in an instance where we re-scale all capacities to integer values and obtain an integral flow of the same value as maximum weight matching in $G$. \SR{An integral flow does not mean that we will find a matching!!
%
% Also to use this result we need to change the valuations to rational as opposed to reals.
% }

\textbf{Graphical representation of house allocation.}   Given an instance $\langle N, H, V \rangle$, we construct a bipartite graph $G = (N \cup H, E)$ such that given $i\in N$ and $ h\in H, (i, h) \in E$ if and only if $v_i(h) >0$. We call $G$ as the valuation graph.
When the valuations are not binary, we additionally construct an edge weight function $wt$ defined as: $wt((i,h)) = v_i(h)$ for $i\in N$, $ h\in H$, and $ (i, h) \in E$.
Given an allocation $A$ of $\langle N, H, V \rangle$, we define the \emph{matching corresponding to allocation $A$ in $G$} as the set of vertex disjoint edges $\{(i,h)\in E \mid i \in N, h\in H, \text{ and } A(i) =h \}$.

\section{Material Omitted From Section~\ref{sec:EF}} \label{app:EF}
\subsection{Binary Valuations}\label{app:EF:bin}
When the valuations of agents towards houses are binary, there is a polynomial time algorithm that can return an envy-free allocation of maximum size (or empty) by computing an ``envy-free matching'' on the bipartite graph induced by the house allocation instance \cite{aigner2022envy}.
However, as we illustrated in \cref{ex:tradeoff}, while an envy-free maximum size `matching' of \citet{aigner2022envy} may return `empty', an envy-free allocation of maximum size could be of larger size.

The key difference between our approach is that we allow allocations along zero edges (aka houses that are valued zero and are not adjacent).
%
% For completion, we provide additional structural insights and the detailed proof in \cref{app:EF}.
% \HH{Erel's paper}

% \HH{Erel's paper: max welfare could be of size 10; however, they are looking for EF matching of max size.}

% We start by providing additional structural results from \citet{aigner2022envy} and show how to prove \cref{prop:EFmaxSize} using them. 

 We start by providing some structural results in binary instances. The next proposition consolidates two key ideas presented in Theorem 1.1 (e) and Theorem 1.2 of  \cite{aigner2022envy} on finding envy-free matchings. We, then, use it to prove \cref{prop:EFmaxSize}.

 \begin{proposition}[\protect{\cite{aigner2022envy}}]\label{prop:aigner2022envy}
 Every bipartite graph $G=(N \cup H, E)$ admits a unique partition of $N = N_S \cup N_L$ and of $ H = H_S \cup H_L$ such that 
  every envy-free matching in $G$ is contained in $G[N_L, H_L]$, 
 a maximum matching in $G[N_L, H_L]$ is a maximum envy-free matching in $G$, and it can be computed in polynomial time.
 \end{proposition}

% \begin{proposition}[\cite{aigner2022envy}]\label{prop:aigner2022envy}
% Given a bipartite graph $G=(N \cup H, E)$, the following holds true: 
% \begin{enumerate}[(i)]
%     \item\label{{it:1_prop:aigner2022envy}}  [Theorem 1.1 (e).\cite{aigner2022envy} \HH{it may be better not to reference this way; instead we can mention the precise steps right before (see above)}] Every bipartite graph $G=(N \cup H, E)$ admits a unique partition of $N = N_S \cup N_L$ and of $ H = H_S \cup H_L$  satisfying the following condition: %\HH{notation update to $N$ and $H$.}
%  every envy-free matching in $G$ is contained in $G[N_L, H_L]$.
%  \item\label{it:2_prop:aigner2022envy}[Theorem 1.2.\cite{aigner2022envy}] A maximum matching in $G[N_L, H_L]$ is the maximum envy-free matching in $G$.
% \end{enumerate}
%  \end{proposition}
%\HH{We should mention this in the preliminaries. Also, note that we are avoiding the use of `cardinality'. }

We will use \cref{prop:aigner2022envy} to prove the following proposition.

% \propEFmaxSize*

% \begin{restatable}{proposition}{propEFmaxSize}
% \label{prop:EFmaxSize}
%     Given a binary instance, a \msEF can be computed in polynomial time. 
% \end{restatable}

% \HH{Also, let's not introduce new notation, or at least explain them, e.g. $X_S \dot{\cup}$}

% \begin{restatable}{proposition}{propEFmaxSize}
% \label{prop:EFmaxSize}
%     Given a binary instance, an envy-free allocation of maximum size can be computed in polynomial time. 
% \end{restatable}
% The following proposition extends the result of \cite{aigner2022envy}; additional discussion and proofs are relegated to \cref{app:EF:bin}.

% We show the following using results from \cite{aigner2022envy} as we explain in \cref{app:EF:bin}.

\begin{restatable}{proposition}{propEFmaxSize}
\label{prop:EFmaxSize}
    Given a binary instance, an envy-free allocation of maximum size can be computed in polynomial time. 
\end{restatable}
\begin{proof}
Given a binary instance $\langle N, H, V \rangle$, we start by constructing a bipartite graph $G = (N \cup H, E)$ such that given $i\in N$ and $ h\in H, (i, h) \in E$ if and only if $v_i(h) = 1$. We do the following:
    \begin{enumerate}[(1)]
        \item\label{it:msEF1} Find a envy-free matching $M$ by using \cref{prop:aigner2022envy}. 
        Add all (agent, house) pairs in $M$ to $A$.
        \item\label{it:msEF2} While there exists an unassigned (agent, house) pair $(j,h)$ in $A$ such that each agent $i$ with valuation $v_i(h)>0$ is matched in $M$,  we add the pair $(j,h)$ to $A$. 
    \end{enumerate}

We show that $A$ is envy-free.
By \cref{prop:aigner2022envy}, $M$ is envy-free. Therefore, no envy is created in step~\eqref{it:msEF1}. Let $(j,h)$ be a pair assigned in the step~\eqref{it:msEF2} of the algorithm. No agent can be envious of $j$ since each agent $i$ that likes the house $h$ is assigned to another house it likes in step~\eqref{it:msEF1}, i.e., for an agent $i$, if $v_i(h) >0$, then we have that $v_i(A(i)) = 1$. Since this holds true for each iteration of step~\eqref{it:msEF2}, allocation $A$ is envy free. 

% \HH{the argument is incomplete: We do not just arbitrary add the remaining edges; we only add those houses that remain unmatched and are either 1) not liked by any agent or 2) not liked by any unmatched agent. In other words, we can still add the houses to some other unmatched agents as long as ALL agents incident to this house are already matched, or no one is incident.}

Now we show that $A$ is of maximum size. Towards this, first note that   every EF allocation matches only the houses in $H_L$ (from \cref{prop:aigner2022envy}). Thus it suffices to show that maximum possible houses in $H_L$ are allocated in $A$. Observe that each house in $H_L$ satisfies the premise for step~\eqref{it:msEF2} since $M$ is a maximum matching in $G[N_L \cup H_L]$. Therefore, step~\eqref{it:msEF2} of our algorithm allocates houses in $H_L$ as long as there is an unassigned agent. Therefore, $A$ has maximum size. Thus, $A$ a \msEF.
\end{proof}

 The above algorithm is based on binary bipartite matchings, and thus, it fails to work when we allow for more expressive preferences beyond binary instances. Nonetheless,
we develop a polynomial time algorithm to find maximum size allocations with zero envy in the next section.

\subsection{Weighted Instances}
% \HH{Given the bounded envy algorithm, this section needs major updating. The \cref{alg:algorithmenvyt} is already fully explained in the main body. Similarly, the proofs here need substantial polishing.}

% We begin by defining some notations that we will use in this section. Given a weighted instance $\langle N, H, V \rangle$,  we show that by repeatedly removing a minimal Hall violators of the induced graph, we can compute a \msEF.
%
%% we  say that an house $h \in H$ is in $\topp(i)$, \emph{top preferred houses of an agent $i \in N$}, if $v_i(h) \in \argmax_{h' \in H}{v_i(h')}$.
%% \HH{do we need `top' given that now in Alg 1 we are defining $h_{max}$?}
%
%
\lemhalldeletion*
% \begin{restatable}{lemma}{lemhalldeletion}\label{lem:halldeletion}\label{lem:prop-algmwEF}
%     Given a weighted instance, every house removed by \cref{alg:algorithmenvyt} cannot be contained in an envy-free allocation.
% \end{restatable}
\begin{proof} 
Suppose for contradiction that there exists a house $h$ that is removed by \cref{alg:algorithmenvyt} that can be assigned to agent $i$ under an envy-free allocation $A$.
If $h$ is deleted in \cref{alg:algorithmenvyt}, then it must be included in a Hall violator $X$ - which means all agents in $X$ cannot be assigned to the houses in $X$. 
There must then exist an agent $j \in X$ that does not receive its most preferred house $h$ - or any house of equal value - despite it being assigned to some other agent.
All houses $h'$ with a positive valued edge to agent $j$, added to $E$ in later steps of \cref{alg:algorithmenvyt} must satisfy $v_j(h') < v_j(h)$, since we add edges in decreasing order of preference. Clearly any allocation where $h$ is assigned leaves $j$ envious, regardless of the house agent $j$ may later receive from a maximum size allocation on $G$.

Thus, given a Hall violator $X$ in the graph $G$, no house $h \in X$ can be contained in an envy-free allocation.
\end{proof}
Next we present the complete proof of \cref{thm:algorithmenvyt}.

% \begin{restatable}{theorem}{thmalgorithmenvyt}\label{thm:algorithmenvyt}
%     Given a weighted instance, \cref{alg:algorithmenvyt} returns an envy-free maximum size allocation. %$A$ such that for every pair of agents $i,j\in N$, $\envy_{i,j}(A) = 0$.
% \end{restatable}
\thmalgorithmenvyt*
\begin{proof}
First, note that \cref{alg:algorithmenvyt} runs in polynomial time because every component of the algorithm including finding a inclusion-minimal Hall violator~\cite{GSV19envy,aigner2022envy} and computing a maximum size bipartite matching runs in time polynomial in $n$ and $m$.
Therefore, it suffices to prove that 
i) every house removed by the algorithm cannot be contained in any envy-free allocation, and
ii) a maximum size bipartite matching on the induced graph returns a maximum size allocation among all envy-free allocations.

Let $A$ denote the allocation returned by the algorithm. Statement (i) immediately follows from \cref{lem:halldeletion}. Statement (ii) follows from the observation that \cref{alg:algorithmenvyt} finds a maximum size bipartite matching in the induced graph where every agent only has edges to its most preferred houses in the remaining instance.  %, no further edges can be added, and no Hall violators exist.

%We consider the two types of agents that are allocated a house in $A$. 
First let us consider the agents that receive a house that they value positively.
Note that a maximum size bipartite matching in $G$ first finds a maximum matching in $G$.
In a maximum-size matching on $G$, an assigned agent is given a house it values most in $G$ and no further assignment to a positively valued house is possible. %Further, any unassigned agents will not envy any agent allocated in $A$ because their valuations for the assigned houses must be zero.
%
% Moreover, there is no Hall violator in $G$.
% Then, for each agent $i$ allocated some house $h$, it must be the case that 
% $v_i(h) \geq v_i(h')$, for all $h' \in H$ in the graph $G$. Thus, if agent $i$ is left unassigned in a maximum size matching in $G$, then for each $h \in H$ assigned in $A$, we have that $v_i(h) = 0$.
%

Next consider the agents that receive a zero valued house in $A$. A maximum size bipartite matching in graph $G$ will assign the remaining agents (that cannot be assigned in a maximum size matching in $G$) to zero valued houses while such a house is  available. 
Thus the algorithm assigns maximum number of agents to their positively valued houses in $G$ and maximum number of agents to their zero valued houses in $G$.

Recall that since there is no Hall violator in $G$ and a house deleted from $G$ is never assigned, by \cref{lem:halldeletion} the maximum size bipartite matching does not create envious agents. Thus, it returns the maximum size envy-free allocation on the instance.
\end{proof}

Given a weighted instance $\langle N, H, V \rangle$ we first show how to find a maximum welfare allocation among all the EF allocations in polynomial time using \cref{alg:algorithmenvyt}.

\begin{lemma}\label{lem:WeightedEFmaxWelfare}
    An allocation $A$ returned by \cref{alg:algorithmenvyt} is a maximum utilitarian welfare EF allocation in $\langle N, H, V \rangle$. 
\end{lemma}

\begin{proof}
    Allocation $A$ is envy-free by design. We show that it has maximum \sw{} among the EF allocations.
    Let $N_1 \subseteq N$ denote the set of agents that has non-zero value for some house in $G$ when $A$ is returned by \cref{alg:algorithmenvyt}.
    First we show that each agent $i \in N_1$ receives their highest valued house in $G$. %or a zero valued house according to \cref{def:maxbiaprtite}.
    Then, we show that no higher valued house can be assigned to any agent. Thus, we show that an agent cannot receive a higher valued house in an EF allocation. Consequently, $A$ must have maximum welfare among the EF allocations.
    
    Clearly, since there are no Hall violators in $G$, we have that allocation $A$ allocates a house to each agent in $N_1$. Moreover, if $(i,h)$ is an edge in $G$, then, since the allocation is EF, $h$ is a highest valued house for $i$ among the houses that were not deleted. Thus, if an agent $i \in N_1$ receives a house $h$ in allocation $A$, then house $h$ is a highest valued house for $i$ in $G$.

    Next we show a higher valued house cannot be assigned in to an agent in any EF allocation. Suppose that an agent $i$ receives a house $h$ (recall, $h=\emptyset$ when $i$ does not receive a house) in $A$. Let $h'$ be a house that has higher value for $i$. 
    %We show that $h'$ cannot be assigned in any EF allocation.
    Since we add the edges in decreasing order of value in $G$ and $v_i(h')>v_i(h)$, the house $h'$ must have been added to $G$ and removed by the algorithm. Then, using \cref{lem:halldeletion}, we have that $h'$ cannot be assigned in any EF allocation. Finally, since $A$ is maximum size allocation in \cref{alg:algorithmenvyt}, by \cref{def:maxbiaprtite} no agent can receive a higher valued house.

    Thus, welfare of $A$ is maximum among the EF allocations. 
\end{proof}

Observe that if there exists an EF allocation among the ones with maximum \sw{}, then we can find that using  \cref{lem:WeightedEFmaxWelfare} by checking if the allocation has maximum welfare.

 \thmWeightedEFmaxWelsize*

\begin{proof}
% Given an weighted instance $I = \langle N, H, V \rangle$, we 
% show that the allocation $A$ returned by \cref{thm:algorithmenvyt} when $t$ is set to zero is a maximum welfare envy-free allocation. Then we use this allocation to extend to an EF allocation of maximum size.
From \cref{thm:algorithmenvyt} we have that $A$ is an EF matching. In \cref{lem:WeightedEFmaxWelfare}, we prove that $A$ is of maximum welfare among the envy-free allocations. If $\sw(A)$ is the same as the maximum utilitarian welfare of an allocation in the given instance $ \langle N, H, V \rangle$, then we return $A$; otherwise we return `No'.  
Since $\sw(A)$ is the maximum welfare achieved by any EF allocation, the correctness of this step follows.
Thus, we show the proposition.
\end{proof}

\section{Material Omitted From  Section~\ref{sec:utilwel}}\label{app:utilwel}

\subsection{Minimum \#Envy}\label{app:subsec:memw}
 Given an instance $\langle N, H, V \rangle$, we  construct an instance  of \mcmw{} that we later use to solve \memw{}. 

 \paragraph{\textbf{Algorithm description}}
We construct a bipartite graph $G(N \cup H',E)$ from $I = \langle N, H, V \rangle$, on vertex set $N \cup H'$ where the set $H'$ is constructed by adding a set of $n$ dummy houses to the set of houses $H$, i.e., $H' = H \cup \{h^i \mid i \in N\}$.
  For an agent $i \in N$ and house $h \in H'$, the pair $(i,h) \in E$ if and only if house $h \in H$ and $v_i(h) >0$, or $h \in H'\setminus H$. We define a cost function $c: E \mapsto \mathcal{R}$ on edges of $G$. For ease of exposition we assume $v_i(h) = 0$ for each agent $i$ and a dummy house $h\in H'\setminus H$. Before we define the cost function we scale the valuations $V$ such that for each agent $i$ and house $h$, if we have that $v_i(h) >0$, then $v_i(h) \geq 1$. Now we define cost function $c$ as follows:
    \begin{itemize}
        \item for each agent $i$ and house $h \in H'$ we define $c(i,h) = 0$ if $h$ belong to the set of houses with highest value for agent $i$, and otherwise $c(i,h) = 1$. We call this as the envy component of cost and write it as $c_{envy}(i,h)$.
        \item for each edge $(i,h)$ in $G$ such that $h \in H$, we add $-v_i(h)\cdot L$ to $c(i,h)$ where $L = n+1$. We call this as the welfare component of cost and write is as $c_{sw}(i,h)$.
    \end{itemize}
    Thus, the cost of an edge $c(i,h)$ in $G$ is $c_{sw}(i,h) + c_{envy}(i,h)$, as given in \cref{alg:WeightedminEmaxWelfare}.
    % below:
    % \begin{equation*}
    %         c(i,h) = \begin{cases}
    %              -v_i(h)\!\cdot\!L& \textit{if $h \in H$ is a highest valued house for $i$,}\\
    %              -v_i(h)\!\cdot\!L\!+\!1\!& \textit{if $h \in H$ is not a highest valued house for $i$}\\
    %              1 & \textit{if $h \in H'\setminus H$}
    %         \end{cases}
    %         % & - v_i(h)\cdot L, \textit{~ if $h \in H$ is a highest valued house for $i$,}\\
    %         % =& - v_i(h)\cdot L +1, \textit{~ if $h \in H$ is not a highest valued house for $i$}\\
    %         % =& 1, \textit{~ if $h \in H'\setminus H$}.
    % \end{equation*}
     
     This completes the construction of the graph $G$. Finally, we return a \mcmw matching in $G$ as a min \#envy max \sw{} allocation.

\begin{algorithm}[t]\small
        \caption{a \memw{} allocation}
    \label{alg:WeightedminEmaxWelfare}
    \begin{algorithmic}[1]
    \REQUIRE A house allocation instance $\langle N, H, V \rangle$.
    \ENSURE An allocation of maximum \sw{} that minimizes the number of envious agents.
    \STATE Create a bipartite graph $G = (N \cup H', E)$ s.t. $H' = H \cup \{h^i \mid i \in N\}$ and $(i,h) \in E$ if $v_i(h) >0$, or $h \in H'\setminus H$ for an agent $i \in N$ and a house $h \in H'$.
    \STATE Let $c: E \mapsto \mathcal{R}$ be defined as follows for an edge $(i,h) \in E$:
    \begin{equation*}
            c(i,h) = \begin{cases}
                 -v_i(h)\!\cdot\!L& \textit{if $h \in H$ is a most preferred house}\\
                 &\hfill\textit{ for agent $i$,}\\
                 -v_i(h)\!\cdot\!L\!+\!1\!& \textit{otherwise}\\
                 %1 & \textit{if $h \in H'\setminus H$}
            \end{cases}
    \end{equation*}
    where $L=n+1$.
    \RETURN a \mcmw matching in $G$.
    \end{algorithmic}
\end{algorithm}

     We begin by observing some properties of cost a matching in $G$. For a matching $M$ in $G$ we write $c(M)$ to denote $\sum_{(i,h)\in M} c(i,h)$ and define $c_{sw}(M)$ and $c_{envy}(M)$ analogously. Note that $c(M) = c_{sw}(M) + c_{envy}(M)$.

    \begin{lemma}\label{lem:weightedmemw-maxc-envy}
        Let $M$ be a perfect matching in $G$. Then $c_{envy}(M) \leq n$.
    \end{lemma}
    \begin{proof}
        Note that for each agent $i \in N$ the cost $c_{envy}(i, h)$ is at most $1$ for all $h \in H'$. Since $M$ is a matching of size at most $n$, $c_{envy}(M) \leq n$.
    \end{proof}
        % \HH{I believe this is what we are saying here:
        % Any matching of agents to houses on $G$ consists of $n$ agents. If all agents are envious, there are at most $n$ envious agents. Thus, $c_{envy}(M) \leq n$.        
        % }\SR{No, we cannot say this. Because we are talking about $c_{envy}(M)$ in $G$, not the number of envious agents.}
        % \HH{two points:% 1) do we need the maximum size assumption here? I think it will still be true,
        % 2) if we further assume max-welfare, I believe we can show that $c_{envy}(M) \leq n-1$}
        %
        % \SR{%(1)For us $M$ is perfect matching on $G$. I changed the problem name to explicitly mention it (advantage of using a macro, I had to change only at one place ;)) 
        % (2)Why do we need to show $n-1$? We need an upper bound on envy for setting the value of $L$.}

    We define an \emph{allocation $A_M$ from a matching $M$} in $G$ as follows: for each edge $(i,h) \in M$ such that $h \in H$, assign house $h$ to $i$ in $A_M$. The rest of the agents remain unassigned in $A_M$. Then the following lemma follows from the definition of $c_{sw}$.
    
    \begin{lemma}\label{lem:weightedmemw-cost-wel}
        Let $M$ be a matching in $G$. Then $\sw(A_M) = - \frac{1}{L}c_{sw}(M)$.
    \end{lemma}
    \begin{proof}
        Recall that for each edge $(i,h) \in M$ such that $h \in H$, cost $c_{sw}(i,h) = -v_i(h)L$ and we assign house $h$ to $i$ in $A_M$. Therefore, for each edge $(i,h) \in M$, we add welfare $-c_{sw}(i,h)/L$ to the utilitarian welfare of $A_{M}$. Thus, $\sw(A_{M}) = \sum_{(i,h) \in M} -c_{sw}(i,h)/L = - \frac{1}{L}\sum_{(i,h) \in M} c_{sw}(i,h) = - \frac{1}{L}c_{sw}(M)$.
    \end{proof}

Finally we are ready to prove the theorem.
\thmWeightedminEmaxWelfare*

\begin{proof}
   %\HH{is this defined?} \SR{Yes}
  To prove the theorem we first observe that \cref{alg:WeightedminEmaxWelfare} runs in polynomial time since construction of $G$ and finding a \mcmw{} matching can be done in polynomial time. Next, we prove the correctness of the algorithm.

    Let $M^{*}$ be a \mcmw in $G$.
    First we show that $A_{M^*}$ is a maximum-welfare allocation. Then we will show that it has minimum number of envious agents.
    
    Suppose that there exists an allocation $M'$ such that $\sw(M') > \sw(M^*)$. Then the following calculations give us a contradiction to the fact that $M^*$ has minimum cost in $G$. The first inequality follows since for any house $h$ that is liked by agent $i$, we have that $v_i(h)\geq 1$.
    \begin{align*}
       \sw(M') -1 &\geq \sw(M^*) \\
       -\sw(A_{M'})+1 &\leq -\sw(A_{M^*}) \\
       -\sw(A_{M'})\cdot L +L &\leq -\sw(A_{M^*})\cdot L
    \end{align*}
    In the next line, we replace welfare by $c_{sw}$ using \Cref{lem:weightedmemw-cost-wel} and replace $L$ by $c_{envy}(M')$. Then the inequality changes to strict since using \cref{lem:weightedmemw-maxc-envy}, envy of any matching is less than $n < L$. 
    \begin{align*}
        -c_{sw}(M')\cdot L +c_{envy}(M') &< c_{sw}(M^*)\cdot L\\
    \end{align*}
    The next line follows since $c_{envy}(M^*)$ is non-negative. 
    \begin{align*}
        - c_{sw}(M')\cdot L +c_{envy}(M') &< c_{sw}(M^*)\cdot L + c_{envy}(M^*)\\
        c(M') &< c(M^*).
    \end{align*}
    The final inequality follows from the definition of cost function. Hence, we get a contradiction. Thus, $A_{M^*}$ has maximum $\sw$. %\HH{I'd suggest including these as comments next to each line. Or, write them line by line.}
   
    Given that $A_{M^*}$ has maximum $\sw$, we show that number of envious agents in $A_{M^{*}}$ is minimum.

    \begin{lemma}\label{lem:minenvyenvious}
        An agent $i$ is envious in the allocation $A_{M^*}$  if and only if it is not assigned to one of its most preferred houses.
    \end{lemma}
    \begin{proof}
Suppose agent $i$ is envious, then clearly it is not assigned to its most preferred house. We prove the other direction.
Suppose $i$ is not assigned to any of its most preferred houses. Let $h$ be a most preferred house for agent $i$.  We prove that $i$ is envious due to $h$ by showing $h$ assigned in $A_{M^*}$, i.e., $h$ is matched in $M^*$. 
We have that $v_i(h) > v_i(M^*(i))$ since $h$ is a most preferred house and $M^*(i)$ is not. %Then, $c_{sw}(i,h) < c_{sw}(i,M^*(i))$. 
Moreover, $c(i,h)= -v(i,h)L $ and $c(i,M^*(i))= -v(i,M^*(i))L +1$ from the construction of the cost function. Thus, $c(i,h)< c(i,M^*(i))$. 
 Therefore, if $h$ is not matched in $M^*$, then replacing $(i,M^*(i))$ by $(i,h)$ in $M^*$ decreases its cost, contradicting the fact that $M^*$ has minimum cost. Thus, we prove that agent $i$ is envious.
    \end{proof}
    
 Therefore, using \cref{lem:minenvyenvious} and from the construction of the cost function, an agent adds $1$ to $c_{envy}(M^*)$ if and only if it is envious.
   % Hence, \#envy of allocation $A_{M^*}$ in $I$ is at least cost $c_{envy}(M^*)$ in $G$. Moreover, if agent $i$ that is matched along an edge $(i,h)$ of $c_{envy} = 0$, then it is matched to a most preferred house, consequently, is not envious. 
   Hence, the cost $c_{envy}(M^*)$ in $G$ is the number of envious agents in $A_{M^*}$ in $I$. 
   
    To complete the proof, we need to show that $M^*$ minimizes the cost $c_{envy}$ in $G$. Suppose towards contradiction there exists a perfect matching $M'$ in $G$ such that $A_{M'}$ has maximum welfare in $I$ and $c_{envy}(M') < c_{envy}(M^*)$ in $G$. Then, from \cref{lem:weightedmemw-cost-wel}, we have that $c_{sw}(M') = c_{sw}(M^*)$ since both $A_{M'}$ and $A_{M^*}$ has maximum \sw{}. Therefore, $c_{sw}(M') + c_{envy}(M')< c_{sw}(M^*) +c_{envy}(M^*)$. Thus, we get $c(M')< c(M^*)$, contradicting the fact that $M^*$ is a minimum cost perfect matching in $G$. 
   Since $M^*$ has minimum $c_{envy}$, allocation $A_{M^*}$ minimizes the number of envious agents among all maximum welfare allocations.
\end{proof}

% \begin{restatable}{proposition}{propWelCompmlessn}\label{prop:wel-comp-m-less-n}
% In a binary instance when $m \leq n$, given a \memw allocation $A$, a \complete allocation $\hat{A}$ can be constructed in polynomial time such that
% \begin{enumerate}
%     \item $A$ and $\hat{A}$ has equal welfare and envy;
%     \item $\hat{A}$ is a \mec allocation.
% \end{enumerate}
% \end{restatable}
\propWelCompmlessn*
\begin{proof}
% Let $A$ be a \memw allocation. 
We begin the proof by constructing a complete allocation $\hat{A}$ in time polynomial in $m$ and $n$. We set $\hat{A} = A$. Then we proceed from $A$ iteratively: add a pair $(i,h)$ to $\hat{A}$ where agent $i$ and house $h$ are unassigned in $\hat{A}$, until $\hat{A}$ is complete.

We show that $\#\envy(\hat{A}) = \#\envy(A)$. %The envy of an unmatched agent does not change as each of its preferred house is matched in $A$ as well.
Let $H_u$ be the set of unassigned houses in $A$. We show that assignment of a house in $H_u$ does not create more envious agents than in $A$. 
First, observe that for each agent $i \in N$ and each house $h \in H_u$, we have $v_i(h) \leq v_i(A(i))$ (recall, $v_i(A(i)) = 0$ when $i$ is unassigned). 
% \M{What is this trying to say??}\SR{We are proving the previous sentence.}
 Otherwise, if $v_j(h) > v_j(A(j))$ for some agent $j$ and $h \in H_u$, then we increase the welfare of $A$ by adding $(j,h)$ to $A$ and removing $(j,A(j))$ from $A$, contradicting the fact that $A$ has maximum welfare. 
Therefore, no agent envy an house $h \in H_u$ that is assigned in the iterative step. Then the envy of each agent remains the same as in $A$ after the iterative step. Thus, $\#\envy(\hat{A}) = \#\envy(A)$.
% Additionally, observe that since $m \leq n$, each house must be matched. 

Next, we prove that $\sw(A) = \sw(\hat{A})$.
 % Let $A$ and $\hat{A}$ be the \memw and \mec allocations as constructed in the proof of \Cref{prop:wel-comp-m-less-n}.
 % Recall that we showed that $\envy(A) = \envy(\hat{A})$.
 % It remains to show that $\hat{A}$ has the same welfare as $A$.
 It is clear from the construction that $\sw(\hat{A}) \geq \sw(A)$ since we assign more houses in $\hat{A}$. Observe that each agent $i$ that is assigned a house $h$ in $H_u$ receives utility zero; otherwise, we could add $(i,h)$ to $A$ to increase its welfare. Therefore, $\sw(\hat{A}) = \sw(A)$.
 This completes the prove of \eqref{it:wel-comp-m-less-n1}.

\begin{sloppypar}We prove \eqref{it:wel-comp-m-less-n2} by contradiction. Suppose that there exists a \mec allocation $A^*$ such that $\#\envy(A^*) < \#\envy(A)$. %Since the number of envious agents in $A^*$ is less than the number of envious agents in $A$, then let $j$ denote an agent that is envious in $A$ but not in $A^*$. 
Then, let $j$ be an agent that is envious in $A$ but not in $A*$. %\in envy(A) \setminus envy(A^*)$ be an agent that is envious in $A$ and not in $A^*$. 
Then, from the fact that the valuations are binary, we have $v_j(A^*(j))=1$ since all houses are assigned in $A^*$ as $m\leq n$ and $j$ is not envious in $A^*$. Additionally, $v_j(A(j))=0$ since $j$ is envious in $A$.
Therefore, each agent that is envious in $A$ but not in  $A^*$ adds one to $\sw(A^*)$ and zero to $\sw(A$). The converse holds true as well. 
If $v_i(A(i))=1$ and $ v_i(A^*(i))=0$ for an agent $i$, then $i$ must be envious in $A^*$ but not in $A$ due to binary valuations. % \in envy (A^*) \setminus envy(A)$. 
Therefore, more agents add one to the welfare of $A^*$ than that of $A$ since $A$ has more envious agents than $A^*$. Thus, we get that $\sw(A^*) > \sw (A)$, a contradiction to the fact that $A$ has maximum welfare.\end{sloppypar}
\end{proof}

As a consequence of \Cref{thm:WeightedminEmaxWelfare} and \Cref{prop:wel-comp-m-less-n} we get the following. 
% We immediately get the following. \SR{\cref{cor:minEnvycomplete} follows from Proposition 1~\cite{MadathilMS23}.}
\begin{restatable}{corollary}{corrminEnvycomplete}
\label{cor:minEnvycomplete}
    Given a binary instance, a \mec can be computed in polynomial time when $m \leq n$.
\end{restatable}

\begin{example}[\cref{prop:wel-comp-m-less-n} does not hold for weighted instances]\label{ex:mecWtm-lessn}
Consider three houses and three agents as shown in \cref{fig:wMEMWvsMEC}. The values are shown on edges. The \memw{} allocation shown in red must allocate $h_1$ to $a_1$ with the value of $100$ to satisfy the maximum \sw{} constraint. This allocation results in creating two envious agents (agents 2 and 3). However, The \mec{} allocation (shown in green) has exactly one envious agent.
 
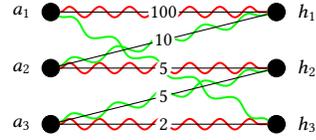
\begin{figure}[t] 
    \centering
    \begin{minipage}{0.45\textwidth}
        \centering
        {\footnotesize
        \begin{tikzpicture}[every node/.style={draw,circle}, fsnode/.style={fill=black}, ssnode/.style={fill=black}, vnode/.style={draw=white, fill=white, inner sep=0pt,outer sep=0pt, midway}]
        
        % the agents
        \begin{scope}[start chain=going below,node distance=5mm]
        \foreach \i in {$a_1$, $a_2$, $a_3$}
          \node[fsnode,on chain] (f\i) [label=left: \i] {};
        \end{scope}
        
        % the houses
        \begin{scope}[xshift=3cm,yshift=0cm,start chain=going below,node distance=5mm]
        \foreach \i in {$h_1$, $h_2$, $h_3$}
          \node[ssnode,on chain] (s\i) [label=right: \i] {};
        \end{scope}
        
        % the edges
        \tikzset{decoration={snake,amplitude=.7mm,segment length=4mm,
                       post length=0mm,pre length=0mm}}

        \draw[decorate, red, line width=0.75] (f$a_1$) -- (s$h_1$);
        \draw[decorate, red, line width=0.75] (f$a_2$) -- (s$h_2$);
        \draw[decorate, red, line width=0.75] (f$a_3$) -- (s$h_3$);

        \tikzset{decoration={snake,amplitude=.7mm,segment length=6mm,
                      post length=0mm,pre length=0mm}}
        \draw[decorate, green, line width=0.75] (f$a_1$) -- (s$h_3$);
        \draw[decorate, green, line width=0.75] (f$a_2$) -- (s$h_1$);
        \draw[decorate, green, line width=0.75] (f$a_3$) -- (s$h_2$);

        \draw (f$a_1$) -- (s$h_1$) node[vnode] {$100$};
        \draw (f$a_2$) -- (s$h_2$) node[vnode] {$5$};
        \draw (f$a_2$) -- (s$h_1$) node[vnode] {$10$};
        \draw (f$a_3$) -- (s$h_2$) node[vnode] {$5$};
        \draw (f$a_3$) -- (s$h_3$) node[vnode] {$2$};
        
       \end{tikzpicture}}
       \small
        \caption{The \memw{} allocation (shown in red) is not a \mec{} allocation (shown in green) in weighted instances. %\HH{can you fix the weights to look better?}
        }
       \label{fig:wMEMWvsMEC}
    \end{minipage}

\end{figure}

\end{example}

\begin{algorithm}[t]\small
        \caption{Finding a \mec{} allocation when $m \leq n$}
    \label{alg:WeightedminEnvycomplete}
    \begin{algorithmic}[1]
    \REQUIRE A house allocation instance $\langle N, H, V \rangle$ with $m \leq n$.
    \ENSURE A \mec{} allocation $A$.
    \STATE Create a bipartite graph $G_{\topp} = (N \cup H, E_{\topp})$ where $(i, h^{\max}_i) \in E_{\topp}$ if and only if $h^{\max}_i \in \argmax_{h\in H}  v_{i}(h)$
    \STATE Let $M$ be a maximum size matching in $G_{\topp}$.
    \STATE Initialize $A=M$
    \WHILE{there exists a unassigned  house $h$ in $A$}
        \STATE Allocate $h$ to an unassigned agent in $A$.   
    \ENDWHILE
   \RETURN Allocation $A$.
    \end{algorithmic}
\end{algorithm}

Observe that each house is assigned in every complete allocation when $m \leq n$. Thus, if an agent is assigned to its most preferred house, then it cannot be envious of others. On the other hand, each agent that doesn't receive its most preferred house (or one from such set) is envious of some agent. This gives us the following characterization.

% \SR{I put it inside the theorem since it is only pertinent to this theorem.}

% \HH{Can we use something other than observation?}\SR{claim/lemma?}
% \HH{yes, but claim or lemma require proof.}\SR{adding a proof now}
% \HH{if no proof is required, then mention it prior to the Theorem without any proof.} \SR{Its a easy proof. I am not sure if it is completely obvious}
% \HH{then make it a lemma or claim.}\SR{okay}
% \SR{Do you feel the following proof is obvious?}

\begin{restatable}{claim}{obsWeightedNonTopisEnvy}\label{obs:weightedNonTopisEnvy}
When $m \leq n$, in any complete allocation an agent is not envious if and only if it is assigned to one of its most preferred houses.
% an agent is not envious in a complete allocation if and only if it is matched to one of its top preferred houses when $m \leq n$.
\end{restatable}
% \begin{proof}
% Clearly, if an agent is matched to its most preferred house, then it is not envious. For the other direction, observe that since $m\leq n$, each house is assigned in a complete matching. Therefore, if an agent is not assigned to any of its most preferred house, then it envious each agent that is assigned to its most preferred house. 
% \end{proof}
Based on the above characterization, we design \cref{alg:WeightedminEnvycomplete}.
\paragraph{\textbf{Algorithm description}} In \cref{alg:WeightedminEnvycomplete}, given an weighted instance $\langle N, H, V \rangle$, we construct an unweighted  bipartite graph $G_{\topp} = (N \cup H, E_{\topp})$ such that given $i\in N$ and $ h\in H$, we have $(i, h) \in E_{\topp}$ if and only if $h$ is a most preferred houses of agent $i$.
We find a maximum size matching $M$ in the graph $G_{\topp}$. We extend $M$ to construct a complete allocation $A$ as follows: initialize $A= M$; until $A$ is complete, allocate an unassigned house to an unassigned agent in $A$. 

% We prove the correctness of \cref{alg:WeightedminEnvycomplete} in the following theorem. We defer the proof to \cref{app:subsec:memw}.

We prove the correctness of \cref{alg:WeightedminEnvycomplete} in the next theorem.
\thmWeightedminEnvycomplete*

\begin{proof}
Let $A$ be the allocation produced by \cref{alg:WeightedminEnvycomplete}.
  We show that $A$ has minimum envy among the set of complete allocations of $\langle N, H, V \rangle$.

% Now we prove the correctness of the algorithm.
Suppose that there is a complete allocation $A^{*}$ with $\#\envy(A^{*}) < \#\envy(A)$. %Let an agent $i \in \envy (A) \setminus \envy(A^{*})$.
 We obtain a matching $M^{*}$ from $A^*$ as follows: $M^{*}=\{(i, A^{*}(i)) \mid i \in N \textit{ and $i$ is envy-free in $A^*$}\}$.
 We show that $M^{*}$ is a matching in $G_{\topp}$ and has more edges than $M$, a contradiction to the fact that $M$ is of maximum size. % matching in $G_{\topp}$.
 % Using \Cref{obs:weightedNonTopisEnvy}, we have that $M^{*}_{\topp}$ is a matching in $G_{\topp}$.
 
  First, we show $M^{*}$ is a matching in $G_{\topp}$. By using \cref{obs:weightedNonTopisEnvy} for $A^{*}$, we have  $A^{*}(i)$ is a most preferred house for each envy-free agent $i$. %In other words, each envy-free agent $i \notin \envy(A^{*})$. 
  Thus, from definition of $E_{\topp}$, we have that $(i, A^{*}(i)) \in E_{\topp}$ for each envy-free agent $i$ in $A^*$. Therefore, $M^{*}$ is a matching in $G_{\topp}$. 
  Next, we show size of $M$ is strictly less than that of $M^{*}$. 
   From \cref{obs:weightedNonTopisEnvy}, if an agent is matched in $M$, then 
   it is envy-free in $A$. Thus, the number of edges in $M$ is strictly less than that of $M^{*}$ since number of envy-free agents in $A$ is strictly less than that of $A^*$. Therefore, we get a contradiction to the fact that $M$ is a maximum size matching in $G_{\topp}$ and prove that $A$ is a \mec.
   
   The matching $M$ can be found in time $O(m\sqrt{n})$ using a maximum  size matching algorithm~\cite{hopcroftKarpmatching} and extending $M$ to $A$ takes time $O(m+n)$. %Hence, we conclude the theorem.
  % \SR{In this proof maximum weight sounds bit odd since we are in an unweighted graph. Can we just say maximum matching for unweighted graph? Then we can cite a maximum matching algorithm for running time.}
  % \HH{agreed!}
\end{proof}

% Moreover, \cref{prop:wel-comp-m-less-n} gives the following relation between a \mec and a \memw. \todo{change inconsistent macro}
% \HH{emphasize that the next corollary is stronger than the Proposition.}
% \begin{restatable}{corollary}{corrWelCompmlessn}
% \label{corr:wel-comp-m-less-n}
%     Given a binary instance where $m \leq n$, every \mec matching has equal envy and welfare as any \memw matching. 
%     % \SR{Is the welfare part true for every \mec?}
    
%     % \HH{Perhaps this can come before Theorem 1 as a proposition, and theorem 1 becomes a Corollary.} \SR{I made both as corollary since neither require a proof.}
% \end{restatable}

% \begin{proof}
%     Let $A$ and $\hat{A}$ be the \memw and \mec allocations as constructed in the proof of \Cref{prop:wel-comp-m-less-n}. Recall that we showed that $\envy(A) = \envy(\hat{A})$. It remains to show that $\hat{A}$ has the same welfare as $A$. It is clear from the construction that $\sw(\hat{A}) \geq \sw(A)$. Observe that each agent that is assigned a house in the iterative step to construct $\hat{A}$ from $A$, get utility zero. Hence, $\sw(\hat{A}) = \sw(A)$.
% \end{proof}

\subsection{Minimum Total Envy}\label{app:subsec:mtew}

\begin{example}[Example showing \memw{} does not imply \mtemw{} even in binary instances] \label{ex:envy_totalenvy_binary}
    Consider the instance given in \cref{ex:tradeoff}. The allocation shown in blue, that is, $(a_1, h_1)$ and $(a_2, h_2)$ maximizes the welfare and leaves only two agents envious $\{a_3, a_4\}$. Similarly, allocation $(a_3, h_1)$ and $(a_4, h_2)$ maximizes the welfare and leaves two agents envious.
    Yet, the former has a total envy of 2, while the latter has a total envy of 3.
\end{example}

% \begin{restatable}{theorem}{thmminTotalEnvyMaxWelfare}
% \label{thm:minTotalEnvyMaxWelfare}
%     Given a weighted instance, a \mtemw{} allocation can be computed in polynomial time.
% \end{restatable}

We find a \mtemw{} allocation by finding a \mcmw{} allocation in an appropriately constructed graph. 
The construction is similar to the one finding \memw{} allocation with a different cost function as given in \cref{alg:minTotalEnvyMaxWelfare}.

\begin{algorithm}[t]\small
        \caption{A \mtemw{} allocation}
    \label{alg:minTotalEnvyMaxWelfare}
    \begin{algorithmic}[1]
    \REQUIRE A house allocation instance $\langle N, H, V \rangle$.
    \ENSURE An allocation of maximum \sw{} that minimizes the total envy of agents.
    \STATE Create a bipartite graph $G = (N \cup H', E)$ s.t. $H' = H \cup \{h^i \mid i \in N\}$ and $(i,h) \in E$ if $v_i(h) >0$, or $h \in H'\setminus H$ for an agent $i \in N$ and a house $h \in H'$.
    \STATE Let $H^{\max}_i$ be the set of most preferred houses in $H$ for agent $i$.
    \STATE Define $c: E \mapsto \mathcal{R}$ for an edge $(i,h) \in E$ as follows:
   
    \begin{equation*}
     c(i, h)  = \begin{cases}
         -v_i(h)\cdot L,%& \\
         \hfill\textit{\hspace{0.35\linewidth} if } h \in H^{\max}_i&\\[1ex]
        %-v_i(h)\cdot L+\sum_{\bar{h} \in H} v_i(\bar{h}), &\textit{ ~~if } h \in H'\setminus H\\
       -v_i(h)\cdot L+ \sum_{\bar{h} \in H}\max\{v_i(\bar{h}) - v_i(h), 0\}&\\
       \hfill\textit{\hspace{0.5\linewidth}otherwise}&.
     \end{cases}
\end{equation*} 
%\HH{the equation is out of box.}
where $L = \sum_{i\in N}\sum_{\bar{h} \in H} v_i(\bar{h}) +1$.
    \RETURN a \mcmw matching in $G$.
    \end{algorithmic}
\end{algorithm}

 The proof will be analogous to the proof of \cref{thm:WeightedminEmaxWelfare}. We being with some properties of a minimum cost perfect matching in $G$. 

    \begin{lemma}\label{lem:weightedmtemw-maxc-envy}
        Let $M$ be a perfect matching in $G$. Then $c_{envy}(M) < L$.
    \end{lemma}
    \begin{proof}
        Note that for each agent $i \in N$ the cost $c_{envy}(i, h)$ is at most $\sum_{\bar{h} \in H} v_i(\bar{h})$ for any house $h$. Therefore, for all the agents in $N$ the total envy component of the cost $c_{envy}(M) \leq \sum_{i\in N}\sum_{\bar{h} \in H} v_i(\bar{h})$. Since $L = \sum_{i\in N}\sum_{\bar{h} \in H} v_i(\bar{h})+1$, we have that  $c_{envy}(M) < L$.   
    \end{proof}

    The next lemma follows from the definition of $c_{sw}$ and the proof is the same as \cref{lem:weightedmemw-cost-wel}.
    
    \begin{lemma}\label{lem:weightedmtemw-cost-wel}
        Let $M$ be a matching in $G$. Then $\sw(A_M) = - \frac{1}{L}c_{sw}(M)$.
    \end{lemma}

    Now we are ready to proof the correctness of \cref{alg:minTotalEnvyMaxWelfare}.

\thmminTotalEnvyMaxWelfare*
\begin{proof} %\HH{move to the appendix}
%  % The cost function $c$ is defined as follows: $c(i,h) = 0$ for each $h \in H'$ that is $i$'s most preferred house. Otherwise, $c(i,h) = \sum_{h' \in H}\max\{wt(i,h') - wt(i,h), 0\} $.
% For each edge $(i, h) \in E$, we define weight function 
% \begin{align*}
%     wt(i, h) &= v_i(h), & \textit{ ~~if } h \in H\\
%     &= \epsilon &\textit{ if } h \in H' \setminus H,
% \end{align*}
% where $0< \epsilon <  \frac{1}{mn} \min\limits_{i \in N} \min\limits_{h_1,h_2 \in H} |v_i(h_1)-v_i(h_2)|$.
    % For each edge $(i,h)$ such that $h \in H$ in $G$ we define $wt((i,h))= v_i(h)$. For edges $(i, h')$ such that $h' \in H' \setminus H$, we define $wt((i, h')) = \epsilon$.

 % =======for binary=======
%  \begin{align*}
%      c(i, h) & =  0, \textit{ ~~if } h \in H\\
%      & = \sum_{h \in H} v_i(h), \textit{ if } h \in H' \setminus H.
%  \end{align*}
%  %$c(i, h) = 0$ if $h \in H$; and $c(i, h) = \sum_{h \in H} v_i(h)$ if $h \in H' \setminus H$. 
% For each edge $(i, h) \in E$, weight function 
% \begin{align*}
%     wt(i, h) &= 1, \textit{ ~~if } h \in H\\
%     &= \epsilon \textit{ ~~for some } 0< \epsilon << 1, \textit{ if } h \in H' \setminus H.
% \end{align*}
% %$wt(i, h) = 1$ if $h \in H$; and $wt(i, h) = \epsilon$ for some $0< \epsilon << 1$ if $h \in H' \setminus H$.
Let $M$ be a matching in $G$. We define 
\begin{align}
 \nonumber   A_{M}= \{(i, M(i)) \mid M(i) \in H\} \cup \{(i,\emptyset)\mid M(i) \in H' \setminus H\}%\hfill\tag{ALLOC}%\label{eq:tootalenvyalloc}
\end{align}

Let $M^*$ be a \mcmw{} in $G$.
We show that $A_{M^*}$is a \mtemw allocation where $M^*$ is a \mcmw in $G$.
The proof will be analogous to the proof of \cref{thm:WeightedminEmaxWelfare}. 
    First we show that $A_{M^*}$ is a maximum-welfare allocation. Then we will show that it has minimum total envy of the agents.

    The proof to show that $A_{M^*}$ is a maximum-welfare allocation is the same as in \cref{thm:WeightedminEmaxWelfare} except we use Lemmas \ref{lem:weightedmtemw-cost-wel} and  \ref{lem:weightedmtemw-maxc-envy} to reach a contradiction. We show it below for the sake of completeness.
    
    Suppose that there exists an allocation $M'$ such that $\sw(M') > \sw(M^*)$. Then the following calculations give us a contradiction to the fact that $M^*$ has minimum cost in $G$. The first inequality follows since for any house $h$ that is liked by agent $i$, we have that $v_i(h)\geq 1$.
    \begin{align*}
       \sw(M') -1 &\geq \sw(M^*) \\
       -\sw(A_{M'})+1 &\leq -\sw(A_{M^*}) \\
       -\sw(A_{M'})\cdot L +L &\leq -\sw(A_{M^*})\cdot L
    \end{align*}
    In the next line, we replace welfare by $c_{sw}$ using \Cref{lem:weightedmtemw-cost-wel} and replace $L$ by $c_{envy}(M')$. Then the inequality changes to strict since using \cref{lem:weightedmtemw-maxc-envy}, envy of any matching is strictly less than $L$. 
    \begin{align*}
        -c_{sw}(M')\cdot L +c_{envy}(M') &< c_{sw}(M^*)\cdot L\\
    \end{align*}
    The next line follows since $c_{envy}(M^*)$ is non-negative. 
    \begin{align*}
        - c_{sw}(M')\cdot L +c_{envy}(M') &< c_{sw}(M^*)\cdot L + c_{envy}(M^*)\\
        c(M') &< c(M^*).
    \end{align*}
    The final inequality follows from the definition of cost function. Hence, we get a contradiction. Thus, $A_{M^*}$ has maximum $\sw$. 
   
    Given that $A_{M^*}$ has maximum $\sw$, we show that the total envy of agents in $A_{M^{*}}$ is minimum.

    \begin{lemma}
        Total envy of an agent $i$ in the allocation $A_{M^*}$ is \\$c_{envy}(i,M^{*}(i))$.
    \end{lemma}
    \begin{proof}
    To prove the statement we consider each of the three cases in the definition of $c_{envy}$.
    
    If $M^{*}(i)$ is a highest valued house for agent $i$, then $i$ receives its most preferred house and total envy of $i$ is zero. Form definition of $c_{envy}$, we have $c_{envy}(i,M^{*}(i)) = 0$ as $M^{*}(i)$ is a most preferred house.

    If $M^{*}(i)$ is a dummy house, i.e., $M^{*}(i) \in H' \setminus H$, then agent $i$ is not assigned any house in $A_{M^*}$. Then $i$ is envious of all the houses it values more than zero. Thus, value of each house $h \in H$ is added to the total envy of $i$. Therefore, total envy of $i$ is $\sum_{h \in H} v_i(h)$, the same as defined in the cost $c_{envy}(i,M^{*}(i))$.

    Finally, if $M^{*}(i)$ is neither a most preferred house, nor a dummy house, then we claim that agent $i$ is envious of each house that is valued more than $M^{*}(i)$. Let $h$ denote the house $M^{*}(i)$. Let $\bar{h}$ denote a house such that  $v_i(\bar{h}) > v_i(h)$. We show that agent $i$ is envious of $\bar{h}$ in $A_{M^*}$. Towards this we show that $\bar{h}$ is assigned to some agent in $A_{M^*}$. Suppose that $\bar{h}$ is not assigned to any agent in $A_{M^*}$, then we can assign agent $i$ to $\bar{h}$ and increase welfare of $A_{M^*}$, contradicting the fact that $A_{M^*}$ is a max welfare allocation.  Therefore, $\bar{h}$ is assigned in $A_{M^*}$ and agent $i$ envies $\bar{h}$. Thus, it adds $v_i(\bar{h}) - v_i(h)$ to total envy of $i$ for each $\bar{h}$ such that  $v_i(\bar{h}) > v_i(h)$. Observe that this is the same as the cost $c_{envy}(i, M^*(i))$. Thus we complete the proof of the lemma.
    \end{proof}
     Therefore, the cost $c_{envy}(M^*)$ in $G$ is the sum of $c_{envy}(i, M^{*}(i))$ for each agent $i$ since $M^*$ is a perfect matching. Therefore, $c_{envy}(M^*)$ is the same as total envy of agents in $A_{M^*}$ in $I$. 
     
     To complete the proof, we need to show that $M^*$ minimizes the envy component of cost $c_{envy}$ in $G$. Suppose towards contradiction there exists a perfect matching $M'$ in $G$ such that $A_{M'}$ has maximum welfare in $I$ and $c_{envy}(M') < c_{envy}(M^*)$ in $G$. Then, from \cref{lem:weightedmtemw-cost-wel}, we have that $c_{sw}(M') = c_{sw}(M')$ since both $A_{M'}$ and $A_{M^*}$ has the same welfare. Therefore, $c_{sw}(M') + c_{envy}(M')< c_{sw}(M^*) +c_{envy}(M^*)$. Thus, we get $c(M')< c(M^*)$, contradicting the fact that $M^*$ is a minimum cost perfect matching in $G$. 
     Since $M^*$ has minimum $c_{envy}$, allocation $A_{M^*}$ minimizes the total envy of agents among all maximum welfare allocations.
 
\end{proof}

\propWelCompmlessTE*
\begin{proof}
% \HH{this proof can be moved to the appendix.}
First, note that since $m \leq n$, every house must be allocated to some agent. Any \mtemw{} allocation $A^*$ is complete when every house in $H$ is wanted by at least one agent. Otherwise, we can simply create a complete allocation $A$ from $A^*$ as follows. Till there is an unassigned house $h$, assign $h$ to some unassigned agent. Note that completing the allocation does not change the total envy i.e., $\tenvy(A)=\tenvy(A^*)$.
Moreover, it is clear that $A$ and $A^*$ has the same welfare.

We prove that $A$ is a \mtec{} allocation by contradiction. Suppose that $A^*$ is a \mtec{} allocation such that total envy of $A^*$ is less than that of $A$. Then, clearly, welfare of $A^*$ must be less than welfare of $A$.  Then there exists at least one house $h \in H$ such that agent $i =A(h)$, agent $i^* = A^*(h)$, and $v_i(h) > v_{i^*}(h)$.
% Let $h$ be the only such house such that $v_a(h) > v_{a'}(h)$.\SR{why can we assume this?}

We show that there exists an agent $z$ such that $\envy_z(A) > \envy_z(A^*)$ and there is an alternating path from $h$ to $z$ (the edges of the path are alternating between edges of $A^*$ and $ A$). Suppose not. Then consider the longest alternating path $P$ (edges alternating between $A^*$ and $A$) starting with $i^*,h,i, A^*(i)$ and so on. Construct the allocation $A'$ from $A^*$ by changing the assignment of each agent in $P$ according to $A$. Then the total envy of the allocation $A'$ decreases from total envy of $A^*$ since there is no agent $z$ on the path satisfy the condition $\envy_z(A) > \envy_z(A^*)$ and total envy of agent $i$ decreases. Thus, we contradict the fact that $A^*$ has minimum total envy by constructing $A'$. Hence, there is an alternating path  from house $h$ to an agent $z$ such that $\envy_z(A) > \envy_z(A^*)$. Let $z$ denote the first agent on the aforementioned path from $h$ satisfying $\envy_z(A) > \envy_z(A^*)$ and the total envy of the agents path from $h$ to $z$ is less in $A^*$ than in $A$. We denote the path by $P_{hz}$.

However, now we construct an allocation $A''$ by modifying $A$ along the path $P_{hz}$, i.e., for each agent $j$ on the path $P_{hz}$, we change $A''(j)$ from $A(j)$ to $A^*(j)$. The total envy of allocation $A''$ is less than $A$ since the total envy of the agents path from $h$ to $z$ is less in $A^*$ than in $A$. 
Consequently, the sum of the their values for the assigned houses is more in $A^*$ than in $A$.
Therefore, the total welfare contributed by the agents on path $P_{hz}$ is increased. Since the remaining assignments are the same as in $A$, we have that $\sw(A'')$ is more than that of $A$, a contradiction. Thus, the proposition holds.

% Let the amount of total envy generated by $h$ in all agents except $a$ and $a'$ be $r$ units. Since the size of the allocation is fixed, we now have three cases :
% \begin{itemize}
%     \item $a'$ is left unassigned in $A^*$ and utility of $a$ is zero %left unassigned 
%     in $A_\mathsf{MTEC}$ \SR{no agent is unassigned in a complete matching} 
%     and all other allocations are unchanged \SR{why can we assume that the rest of allocation is the same?}:
    
%     the total envy of all other agents - except $a$ and $a'$ - is the same in $A_\mathsf{MTEC}$ and $ A^*$. Further, since $v_a(h) > v_{a'}(h)$, $A^*$ has lesser total envy than $A_\mathsf{MTEC}$, a contradiction.
    
%     \item $a'$ is matched to some house $h'$ in $A^*$ and $a$ is matched to $h'$ in $A_\mathsf{MTEC}$ : since $A_\mathsf{MTEC}$ is does not have maximum welfare, $v_a(h') + v_{a'}(h) < v_a(h) + v_{a'}(h')$. Since all other allocations are left unchanged, the total envy of all other agents for $h$ stays the same - say $r$. Thus the total envy caused by $h$ in $A_\mathsf{MTEC} = r + v_a(h) - v_a(h')$ and in $A^* = r + v_{a'}(h') - v_{a'}(h)$. This implies that any such $A_\mathsf{MTEC}$ has a larger envy than $A^*$, a contradiction.

%     \SR{there are more cases: 
% \item $a'$ is matched to some house $h'$ in $A^*$ and $a$ is matched to $h''$ in $A_\mathsf{MTEC}$ 
%     }
% \end{itemize}

Thus when $m \leq n$, any decrease in welfare corresponds to an increase in total envy. This means that given a \mtemw allocation, we can retrieve a \mtec allocation when $m \leq n$.

% Let $A_{\text{MTEC}}$ be a \mtec/ allocation that assigns some house $h \in H$ to an agent $a' \in N$ such that there exists some agent $a \in N, v_a(A_{\text{MTEC}}) < v_a(h)$ and $v_a(h) > v_{a'}(h)$. $A_{\text{MTEC}}$ thus violates maximum welfare. \\
% We now define an allocation $A'_{\text{MTEC}}$ such that $A'_{\text{MTEC}} = A_{\text{MTEC}} \setminus \{(a', h), (a, A_{\text{MTEC}}(a))\} \cup \{(a, h), (a', A_{\text}$, and claim that $A'_{\text{MTEC}}$ has lesser units of total envy than $A_{\text{MTEC}}$. Let us assume that $h$ generate $k \geq 0$ units of envy in $A'_{\text{MTEC}}$ and $k' > 0$ units of envy in $A_{\text{MTEC}}$. We know $k'$ is positive since $v_a(h) > v_{a'}(h) \geq 0$. If $k = v_{a'}(h) + r$ for some non-negative integer $r$, then $k' = v_{a}(h) + r$, since the total envy of the other agents remains unchanged. By construction, $k > k'$. 

% Thus, for $A_{\text{MTEC}}$ to be \mtec/ no available agent that values a house higher can be left unmatched, in favor of one that values a house less. 
\end{proof}

\section{Omitted proofs from 
Section~\ref{sec:egal}}\label{app:egal}

% \begin{algorithm}[t]\small
%         \caption{Finding an allocation of max \egwel{}}\label{alg:algorithm-egalk}
%     \begin{algorithmic}[1]
%     \REQUIRE A house allocation instance $\langle N, H, V \rangle$.
%     \ENSURE An allocation with maximum egalitarian welfare.
%     \FOR{$k = n$ to $1$}
%         % \STATE Create a bipartite graph $G = (N \cup H, E)$ s.t. there is exists an edge between each agent $i\in N$ and each house $h\in H$ if $v_{i}(h) > 0$ 
%         \STATE Let $v_{\max} \coloneqq \max_{i,h}{v_{i}(h)}$ and $v_{\min} \coloneqq \min_{i,h}{v_{i}(h)}$
        
%         $\triangleright$ \textsc{\footnoresize Looping  all unique values of houses in $H$}
        
%         \FOR{$\beta = v_{\max}$ to $v_{\min}$}
%              \STATE Create a bipartite graph $G_{\beta} = (N \cup H, E)$ s.t. there is exists an edge between each agent $i\in N$ and each house $h\in H$ if $v_{i}(h) \geq \beta$ 
%             % \STATE Remove every edge $(i',h')$ from $G$ if $v_{i'}(h') < \beta$
%             \STATE Let $A \coloneqq$ a maximum size matching of $G_{\beta}$
%             \IF{there exist $k$ allocated agents in $A$}
%                 \RETURN Allocation $A$
%             \ENDIF   
%         \ENDFOR
%     \ENDFOR
%     % \RETURN $\emptyset$
%     \end{algorithmic}
% \end{algorithm} 

% \thmEgalWelfare*

\subsection{Minimum \#Envy}\label{app:egal:minenvy}

\thmminenvykegal*
\begin{proof}
    We prove this by showing a reduction from the problem of finding a \mec{} allocation that is known to be NP-complete.

    Let $I = \langle N, H, V\rangle$ be an instance of the Minimum Envy Complete problem where the goal is to find a complete allocation with at most $z$ envious agents.  %For $k = n$ and envy $=z$,  
    We build an equivalent instance $I' = \langle N, H, V' \rangle$ of the min \#envy max \egwel{} problem. We create the valuation $V'$ as follows: 
    \begin{itemize}
        \item for each agent $i \in N$ and house $h \in H$, we set the value $v'_i(h) = v_i(h) + \beta$, where $0 < \beta <  \min_{(i,h) \in N \times H} v_i(h)$. %, i.e., $\beta$ is a positive constant
    \end{itemize}
    We now show that an allocation with \egwel{} at least $\beta$ for all $n$ agents has envy at most $z$ in the instance $\langle N, H, V' \rangle$ if and only if a minimum envy complete allocation has envy at most $z$ in $\langle N, H, V\rangle$. 
    
Before we show the equivalence, we show the following property. 

    Let $A$ be a \minEnvy{} \egaln{} allocation with $\egwel{}\geq \beta$ in $I'$ with $\#\envy(A) \leq z$. %, for $k = n$. 
    We show that the $\envy(A)$ in $I$ and $I'$ are the same.
    The egalitarian welfare  of $A$ is at least $\beta > 0$., i.e.,  $v_i(A(i)) \geq \beta$ for each agent $i$.  Then,  if an agent $i \in N$ is envious, then $i$ must be envious of an agent who is assigned to a house $h$ such that $v_i(h) > \beta$, since there are more houses than agents, and in $I'$ each agent has positive valued for all the houses %$G'$ is a complete graph.
    Therefore, if an agent is envious in $I'$, then it must be envious in $I$. Thus, each envious agent in $I'$ is envious in $I$ for the allocation $A$.  Further, if agent $i$ is envious due to a house $h$ in $I$, then it is envious in $I'$ since from the construction $\beta \leq v'_i(h) < v_i(h)$ for $(i,h) \in N \times H$.  Hence,  in the allocation $A$,  the number of envious agents in $I$ is the same as in $I'$.  Hence,  $A$ has envy at most $z$ in $I$.
    % \SR{\@ Medha,  the reduction should be for the decision version of the problem not optimization version. So the following line is incorrect:Therefore, is minimized in $G'$, it must also be minimized in $G$ i.e. a minimum envy k-egalitarian matching gives us a minimum envy complete matching.}

    To show the equivalence, first note that a a \minEnvy{} \egaln{} allocation $A$ is a complete allocation and has minimum envy due to the above property. Hence, $A$ is a \mec\ allocation with $\envy(A) \leq z$ in $\langle N, H, V\rangle$.

    For the other direction, let $A'$ be a \mec\ with $\#\envy(A') \leq z$.  We show that $\egwel(A') \geq \beta$ in  $I'$.  This holds since, by construction, $I'$ is a complete graph where every agent has value at least $\beta$ for each house.
    
    % If we have a minimum envy complete allocation in $I$ that uses no $0$-valued house, then the corresponding allocation in $I'$ will have only positive weight edges from $I$ that have weight at least $\beta$ in $I'$ from construction.  If the minimum envy complete allocation in $I$ has any edges of weight zero, then the corresponding edge in $I'$ must have weight exactly $\beta$. 
    
    Thus any complete allocation in $I$ satisfies the \egwel{} threshold $\beta$ in an allocation of size $n$ i.e. any \mec allocation in $I$ must also be a \minEnvy{} \egaln{} allocation in $I'$.
    Since the minimum \#envy complete matching is a \NPH{}~\cite{kamiyama2021envy}, we conclude that minimum \#envy max egalitarian is also \NPH{}.
\end{proof}

\subsection{Minimum Total Envy}\label{app:egal:minTE}
\thmredfromMTEC*
\begin{proof}
    % We prove equivalence below.
    Let $\langle N,H, V\rangle$ be an instance of a minimum total envy complete problem. We use the construction in \cref{thm:minenvyk-egal} to construct an instance of \minTEnvy{} max \egwel{}.  %For $k = n$, we build an equivalent instance of \minTEnvy{} \egalk{} given by the valuation graph $G'$ as follows :
    % \begin{itemize}
    %     \item Set $G' = G$. 
    %     \item for each $e = (u,v) \in E$ we set the weight of $e$ : $wt(e) = wt(e) + \beta$, where $\beta$ is some positive constant.
    %     \item We complete the graph $G'$, and give the newly added edges weight $\beta$.
    % \end{itemize}
    First we show the following property.
    We now show that a minimum total envy allocation $A$ in $\langle N,H, V'\rangle$ of size $n$ and $\egwel(A) \geq \beta$ has the same total envy with respect to the valuations $V$ and $V'$. 
    % Let $A$ be a minimum total envy $k$-egalitarian allocation in $G'$, for $k = n$.
    We use $G$ and $G'$  to denote the two bipartite graphs on vertex set $N\cup H$ where the edges in $G$ and $G'$ are given by the valuations $V$ and $V'$, respectively.
    We know that the egalitarian welfare must be at least $\beta > 0$. It follows that any envious agent $i \in N$ must be envious of some house it values strictly more than $ \beta$. From the construction of the instance we have that each edge that contributes to total envy of an agent $i$ in $G'$ exists in $G$. Thus, the total envy of each agent is the same in $G$ and $G'$. Hence, if the total envy is minimized in $G'$, it must also be minimized in $G$ i.e. a \minTEnvy{} \egalk{} allocation gives us a \mtec\ allocation.

    Now we show the equivalence between the instances. A minimum total envy allocation $A$ in $G'$ of size $n$ and $\egwel \geq \beta$, is a complete allocation in $G$. Due to the above argument, $A$ has minimum total envy. Hence, $A$ is a \mtec{} allocation in $G$.

    For the other direction, given a minimum total envy complete allocation $A$ in $G$, it must have egalitarian welfare $\egwel(A) \geq \beta$ in $G'$ since by construction, $G'$ is a complete graph where each edge has weight at least $\beta$. Thus, $A$ is an allocation in $G'$ of size $n$ and $\egwel \geq \beta$. Therefore, using the above property $A$ is a  \minTEnvy{} \egalk{} allocation in $G'$.
    %  If we have a minimum total envy complete allocation in $G$ that uses no $0$ weighted edges, then the corresponding matching in $G'$ will have only positive weighted edges in $G$ that must have weight $> \beta$ in $G'$. If the minimum total envy complete matching in $G$ has any zero edges, then the corresponding edge in $G'$ must have weight exactly $\beta$. 
    Thus any complete matching in $G$ satisfies the $\beta$-threshold for egalitarian welfare, for size $= n$ i.e. any minimum total envy complete matching in $G$ is a \minTEnvy{} \egalk{} allocation in $G'$.
\end{proof}

\subsection{Minimax Total Envy}\label{app:egal:minimaxE}

\thmmmtenvyegalwelfare*
% \begin{restatable}{theorem}{thmmmtenvyegalwelfare}
%     \label{thm:mmtenvy-egal-welfare}
%     \minmaxTEnvy{} \egaln{} is \NPH{} even when the egalitarian welfare is strictly better than zero.
% \end{restatable}

\begin{proof}
    The construction of the reduction is the same as in ~\cite[Lemma 38]{MadathilMS23} except for the valuation functions. We reduce from the Independent Set problem in cubic graph which is known to be \NPH~\cite{GJ79}. 
    
    We construct an instance $\langle N,H,V \rangle$ of \minmaxTEnvy{} \egaln\ from an instance $G=(X,E),k$ of Independent Set where the goal is to find an independent set of size at least $k$. Let $|X|=\enn$ and $|E| = \emm$.
    
    \textbf{Houses}: For each vertex $x\in X$ create a house $h_x$, we refer to them as vertex-houses and denote by $H_X$. Additionally, we add  $3\emm + \enn-k$ dummy houses to $H$.
    
    \textbf{Agents}: For each vertex $x\in X$ create an agent $a_x$, we refer to them as vertex-agents and denote by $N_X$. For each edge $e\in E$ create three agents $a_e^1, a_e^2,$ and $a_e^3$, we refer to them as edge-agents. That is, $N = N_X \cup \{a_e^1, a_e^2,a_e^3: e\in E\}$.
    
    \textbf{Valuations}: For each $x \in X$, $v_{a_x}(h_x) = \beta+1$ and $v_{a_x}(h) = \beta$ for every house $h \in H \setminus \{h_u\}$. For each $e = \{x,y\} \in E$ and each $i \in [3]$, we set $v_{a_e^i}(h)$ to be $\beta+1$ if $h \in \{h_x,h_y\}$ and set it to be $\beta$ otherwise.
    This finishes the construction.
    
    First observe that value of an agent towards a house is at least $\beta$. Thus any allocation of the instance has egalitarian welfare at least $\beta$. Now the the proof is the same as \cite[Lemma 38]{MadathilMS23}.

    We prove that $G$ has an independent set of size at least $k$ if and only if $\langle N,H,V \rangle$ has an allocation where maximum $\tenvy{}$ of any agent is at most one.
    For the forward direction, suppose that $X' \subseteq X$ is an independent set in $G$ of size at least $k$. Then, in the following allocation has each agent has  $\tenvy{}$ one. We assign the pair $(a_x,h_x)$ for each $x \in X'$ and assign the remaining unassigned agents to dummy houses. The later step is possible since $|X'| \geq k$, so there are at most $\enn-k$ vertex-agents that are not assigned to their corresponding vertex-houses and $3\emm$ edge agents that are unassigned.
    Now, observe that no vertex agent is envious since either it is assigned to the house it values the most, or it is assigned to the second best house and its best valued the house is not assigned to any agent. Moreover, for each edge-agent, the  $\tenvy{}$ is at most one. It follows from the fact that $X'$ is an independent set and so both endpoints of an edge is not present in $V'$. Hence, for an edge $e= \{x,y\} \in E$, we have that at most one of  the houses $h_x$ or $h_y$ is assigned to their corresponding vertex-agent but both are not assigned. Hence, the edge-agents for $a_e^i$ has  $\tenvy{}$ at most one for $i \in [3]$.

    For the other direction, first we show that any allocation that has maximum $\tenvy{}$ at most one can be changed into a nice allocation such that the  $\tenvy{}$ of the maximum envious agent remain the same. Here, an allocation is nice if for every assigned vertex-house $h_x$ it holds that the agent assigned to $h_x$ is $a_x$. Suppose it doesn't hold for some allocation $A$. Then, we add $(a_x,h_x)$ and $(A(h_x),A(a_x))$ to $A$ and delete $(a,A(a_x))$ and $(a(h_x),h_x)$ from $A$. The maximum $\tenvy{}$ experienced by any agent remains the same. If the agent $A(h_x)$ is not an edge-agent such that $v$ is one of the endpoint of the edge, then envy of $A(h_x)$ does not increase.
    
    Otherwise, let the edge be $e = \{x,y\}$ and wlog, $a_e^1$ is assigned to $h_y$ in $A$. The agent $a_e^1$ envies only $a_y$ after the exchange. This follows from the fact that both $a_e^2$ and $a_e^3$ cannot be assigned to $h_x$ that is the only other highest valued house for them. Then at least one of them will envy both $A(h_y)$ and $A(h_x)$, producing a $\tenvy{}$ of two and contradicting the fact that maximum $\tenvy{}$ is one for any agent in $A$. Hence, we assume we have a nice allocation. Now, the set of vertex-houses assigned to vertex-agents form an independent set in $G$. This completes the proof. 
    \end{proof}

    % \SR{Note that in the above reduction there exists an independent set of size $k$ if and only of the \sw{} of the allocation  is $k$. Therefore, we show that minimizing maximum \#envy of an agents in allocation with social welfare at least $k$ is \NPH{}.}

\section{Omitted details from Experiments - Section~\ref{sec:expt}} \label{expt-contd}

\setlength{\belowcaptionskip}{-10pt}

\begin{table*}[t] \centering \small
% \resizebox{\textwidth}{!}{%
\begin{tabular}{@{}cccc@{}}
               \cmidrule(l){2-4} 
               & Utilitarian \sw{} & Min \#envy    & Min total envy \\ \midrule
Binary &   \includegraphics[scale=0.24]{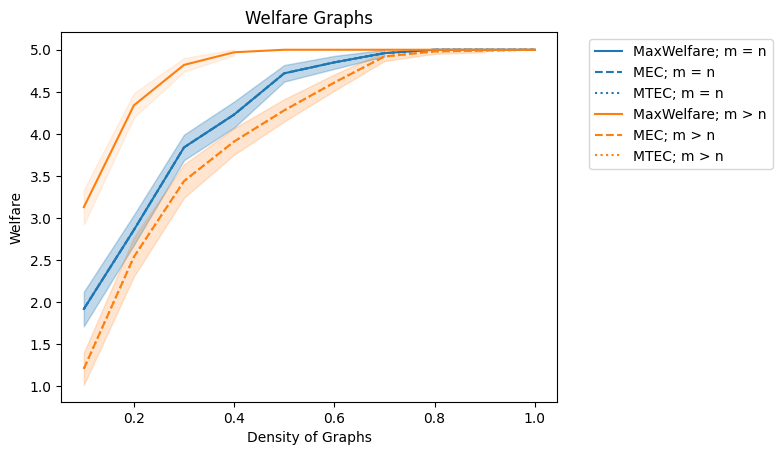}    &   \includegraphics[scale=0.25]{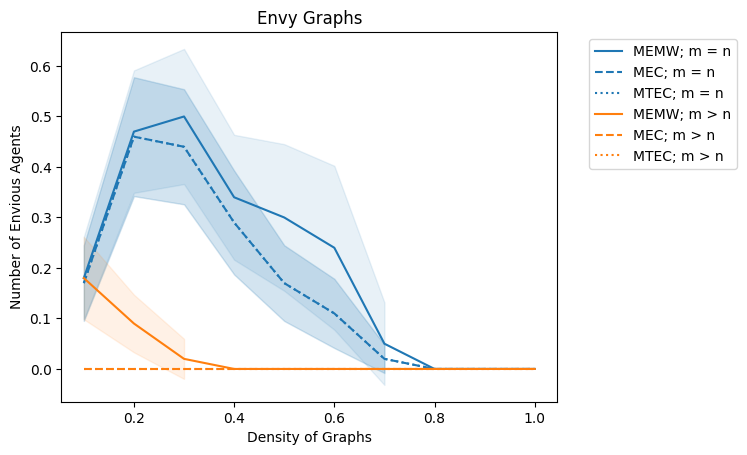}  &  \includegraphics[scale=0.24]{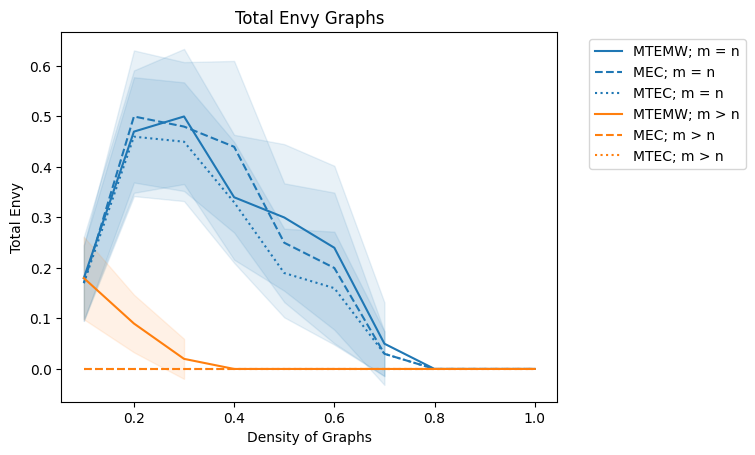}
\\
 
\midrule
\begin{tabular}{c}
   Truncated \\ Borda valuations 
\end{tabular}  &      \includegraphics[scale=0.25]{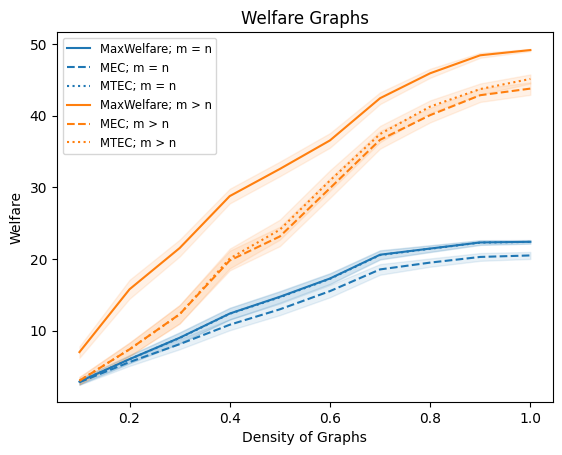}  &   \includegraphics[scale=0.24]{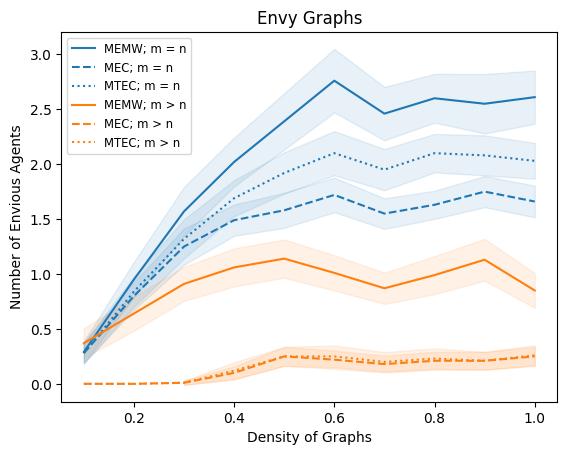}  & \includegraphics[scale=0.25]{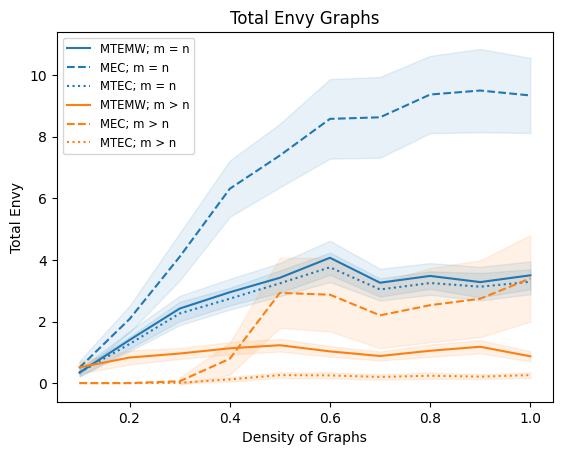} \\
% \midrule
% Weighted &   \includegraphics[scale=0.25]{Figures/Weighted5Agents100Trials/PoF-Welfare.png}    &   \includegraphics[scale=0.25]{Figures/Borda5Agents100Trials/PoF-Envy.png} &        \includegraphics[scale=0.25]{Figures/Weighted5Agents100Trials/PoF-TEnvy.png}           \\ 
\bottomrule
\end{tabular}
%}
\end{table*}
\subsection{Design}

Our code was compiled with Python $3.8$ and we use the NetworkX library to model all instances of the house allocation problem as bipartite graphs. We deconstruct our experiments into instances and trials. We fix the number of agents across all instances and trials to be $n = 5$. 
\subsubsection*{Instances}

An instance $I$ of the problem is defined by the tuple $(m_I, \lambda_I, V_I)$, where $m_I$ is the number of houses, $\lambda_I$ is the probability that an edge $(a, h)$ between $a \in N$ and $h \in H$ exists, and $V_I$ is the type of valuation function. We vary the number of houses ($m_I$) over the set $\{n, 1.6n, 2n\}$ where $n$ is the number of houses, to capture the impact of an abundance of houses on the fairness-efficiency trade-off. 

Similarly, $\lambda_I$ iterates over the $10$ values in the interval $[0.1, 1]$ with step size $= 0.1$. Since all edges are picked independently, 
$\text{Pr}[(a, h) \in G | a \in N, h \in H] = \lambda_I$ implies that $\frac{\#\text{ edges in }G}{\max(\text{possible edges in }G)} = \lambda_I$. Correspondingly, we choose to refer to $\lambda_I$ as the density of the graph.

Lastly, we consider the following three distinct valuation functions:
\begin{itemize}
    \item \textbf{Binary Valuations : } Each existing edge has weight $= 1$ i.e. for all $a \in N, h \in H$ if $\exists (a, h) \in G$ then $v_a(h) = 1$. 
    \item \textbf{Weighted Valuations : } We noticed a great deal of noise when weighted valuations were uniformly assigned over the interval $[1, 100]$ since agent preferences could differ wildly. To minimize this noise while still retaining the fundamental traits that make Weighted Valuations more challenging than Binary Valuations, we assigned each edge $(a, h) \in G$, $v_a(h)= k$ for some random $k \in [\text{MaxDeg}(G), \text{MaxDeg}(G) - \text{Deg}(a) + 1]$ such that $k$ is chosen \textit{without} replacement. Note that $\text{MaxDeg}(G)$ is the maximum degree of a node in $G$ and $\text{Deg}(a)$ is the degree of agent $a$ i.e., the number of houses it has positive valuations for.
\end{itemize}

Thus, the experiments are run over a set of $60$ unique instances, and for each $I$, we generate $100$ random \textit{trials}. 

\subsubsection*{Trials}

A trial 
% \HH{this is an instance specified by a graph, right?} 
is a single randomly generated graph $G$ that satisfies the constraints ($m_I, \lambda_I, V_I$) of a specific instance $I$. Each $G$ was created using the bipartite random graph function in the NetworkX algorithms library, with parameters "$(n, m, \lambda_I)$". Once $G$ is generated, we give the edges $(a, h) \in G$ weights according to $V_I$.

Thus, every \textit{instance} $I$ defines a corresponding family of graphs $S$, and each \textit{trial} $T$ of $I$ generates a graph $G$ that belongs to $S$. 

\subsection{Setup}
% \M{Done} \HH{write setup perhaps?}

Given an instance $I$, in each trial $T$ of $I$ we find the following allocations \memw, \mtemw, \mec, and \mtec using a brute force implementation. For every type of allocation we store the envy, total envy and utilitarian \sw{} generated for that specific trial. Once we have these statistics for all trials and all instances, we plot the mean and $95\%$-Confidence Interval over $100$ trials for all instances. For clarity, we plot each of these metrics on different figures and group them by valuation function. Below, we have the plots for Binary and Truncated Borda Valuations.

\subsection{Observations}
We describe the general trends observed in the plots below:
\begin{itemize}
    \item \textbf{Abundance of Houses : } As the number of houses increases, the \#envy and total envy across all three types of allocations decrease, and the utilitarian \sw{} generated increases. This is consistent with what we expect. An abundance of houses makes it easier to satisfy the demands of all agents since the probability of each agent getting some house allocated to them, i.e. a complete matching, increases. Under binary valuations this increases the number of envy-free agents, and under (positive) weighted valuations this decreases the envy of all envious agents since very few are left unallocated. 
    \item \textbf{Density of the Graph :} Across matchings and valuation profiles, we notice that as the graph grows more connected, i.e., $\lambda_I$ increases, the \#envy of the graph increases correspondingly. Under binary valuations we note that $\lambda ~ 0.3$ is roughly when the graph is connected enough to create large amounts of \#envy and total envy, but disconnected enough to prevent a complete matching. When $\lambda > 0.3$, as agents are increasingly assigned to some house they want, \#envy and total envy decrease correspondingly. Both \#envy and total envy rapidly drop to $0$ under binary valuations, but remain non-zero under truncated borda valuations. The latter can be explained by the fact that truncated borda valuations are assigned without repetition, and an agent's most preferred house being assigned to another, still leaves an agent envious and adding to the total envy of an allocation. Again, this is as expected. Envy only occurs when preferences clash, and this can only happen in reasonably dense graphs. However, since we ensure there are at least as many houses as agents, a further increase in the density of the graph increases the number of completely, or partially, satisfied agents, which explains the decrease in the \#envy and total envy. 
    \item \textbf{\#Envy: } The envy of a \memw{} allocation is consistently higher than that of a \mec{} allocation. This is expected since \mec{} allows for unassigned highly valued houses and assigning houses to agents that do not want them. However, the notable difference between the \#envy generated by these two allocations gives us some insight into the considerable impact this relaxation has on the envy of an allocation. Further, we note that while \mtec has lesser \#envy than \memw{}, which can also be attributed to relaxing the max welfare constraint, \mtec{} also has more \#envy than a \mec{} allocation. This can be explained by the fact that \mtec{} would prefer to \textit{distribute} envy over multiple agents in order to decrease the net total envy, while \mec{} would prefer to concentrate all the total envy generated on one incredibly envious agent. This highlights the fundamental incompatibility between the two measures. 
    \item \textbf{Total Envy: } Under binary valuations, when $m = n$ the Total Envy generated by the three types of allocations is roughly equivalent, with \mtec{} lower bounding the total envy in all instances. However when $m > n$, there is hardly any total envy once the graph is dense enough for a complete matching. However, when we consider the truncated borda valuations, we notice that \mec{} performs significantly worse than both \mtemw{} and \mtec{} allocations. This further hints at total envy being more dependent on the \sw{} of an allocation than \#envy, and thus being somewhat incompatible with envy under weighted valuations.
    \item \textbf{Welfare: }The utilitarian welfare generated by an allocation is significantly higher in \sw{} maximizing allocations, like \memw{} and \mtemw, than in envy or total envy minimizing allocations such as \mec{} and \mtec. This can be the explained by the significant positive impact that allowing highly wanted houses to go unassigned, has on the envy/total envy of an allocation.
\end{itemize}

\end{document}